\documentclass[12pt]{article}

\usepackage{enumitem}
\usepackage{comment}
\usepackage{amsthm}
\usepackage{amsmath}
\usepackage{amssymb}
\usepackage{amsfonts}
\usepackage{graphicx}
\usepackage{mathrsfs}
\usepackage{bm}

\newtheorem{theorem}{Theorem}[section]
\newtheorem{lemma}[theorem]{Lemma}

\newtheorem{openproblem}[theorem]{Open Problem}
\newenvironment{definition}[1][Definition]{\begin{trivlist}
\item[\hskip \labelsep {\bfseries #1}]}{\end{trivlist}}

\newcommand{\Comp}[1]{\ensuremath{\overline{#1}}}
\newcommand{\Rot}[2]{\ensuremath{t_{#1}(#2)}} 
 
\newcommand{\Rev}[1]{\ensuremath{\widetilde{#1}}}
\newcommand{\Back}[1]{\ensuremath{\widehat{#1}}}
\newcommand{\Refl}[2]{\ensuremath{f_{#1}(#2)}} 
\newcommand{\ORefl}[1]{\ensuremath{f_#1}}

\newcommand{\Up}{\textbf{u}}
\newcommand{\Right}{\textbf{r}}
\newcommand{\Down}{\textbf{d}}
\newcommand{\Left}{\textbf{l}}

\newcommand{\Factor}{\ensuremath{\preceq}}
\newcommand{\Prefix}{\ensuremath{\preceq_{\rm{pre}}}}
\newcommand{\Suffix}{\ensuremath{\preceq_{\rm{suff}}}}
\newcommand{\Affix}{\ensuremath{\preceq_{\rm{aff}}}}
\newcommand{\Middle}{\ensuremath{\preceq_{\rm{mid}}}}

\newcommand{\Mirror}{\ensuremath{\preceq_{\rm{mir}}}}

\newcommand{\Bou}[1]{\ensuremath{\mathcal{B}(#1)}}
\newcommand{\D}[1]{\ensuremath{#1^{\circ}}}

\newcommand{\ccNP}{\textrm{\textsc{NP}}}

\author{Stefan Langerman\footnote{Directeur de recherches du F.R.S.-FNRS.}\thanks{D\'{e}partement d'Informatique, Universit\'{e} libre de Bruxelles, ULB CP212, boulevard du Triomphe, 1050 Bruxelles, Belgium, \protect{\texttt{\{stefan.langerman,andrew.winslow\}@ulb.ac.be}}} \and Andrew Winslow \footnotemark[2]}
\title{A Quasilinear-Time Algorithm\\
for Tiling the Plane Isohedrally\\
with a Polyomino}
\date{}

\begin{document}

\maketitle

\begin{abstract}
A plane tiling consisting of congruent copies of a shape is \emph{isohedral} provided that for any pair of copies, there exists a symmetry of the tiling mapping one copy to the other.
We give a $O(n\log^2{n})$-time algorithm for deciding if a polyomino with~$n$ edges can tile the plane isohedrally.
This improves on the $O(n^{18})$-time algorithm of Keating and Vince and generalizes recent work by Brlek, Proven\c{c}al, F\'{e}dou, and the second author.
\end{abstract}

\section{Introduction}
\label{sec:introduction}

The 18th of Hilbert's 23 famous open problems posed in 1900~\cite{Hilbert-1902} concerned \emph{isohedral} tilings of polyhedra where every pair of copies in the tiling has a symmetry of the tiling that maps one copy to the other (see Figure~\ref{fig:isohedrality-ex}).
Hilbert asked for an example of an \emph{anisohedral} polyhedron that admits a tiling, but no isohedral tilings.

\begin{figure}[ht]
\centering
\includegraphics[scale=1.0]{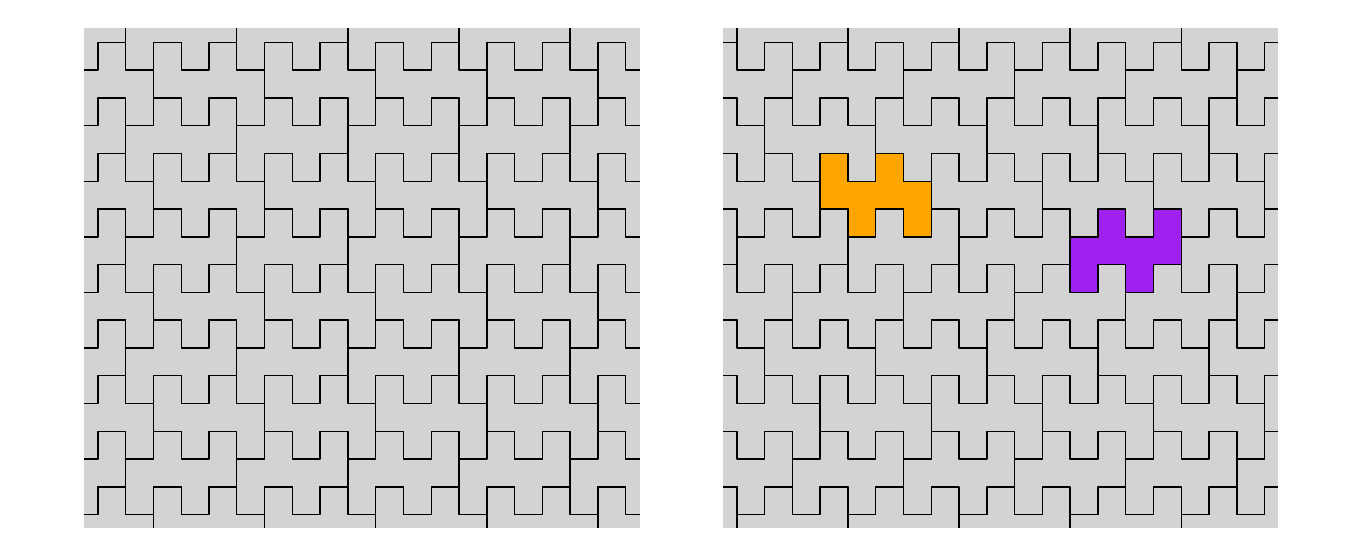}
\caption{Isohedral (left) and anisohedral (right) tilings of a polyomino.
There is no symmetry of the right tiling mapping one colored tile to the other.}
\label{fig:isohedrality-ex}
\end{figure}

Reinhardt~\cite{Reinhardt-1928} was the first to give an example of an anisohedral polyhedron.
Along with this example, Reinhardt also stated that a proof that no anisohedral polygons exist was forthcoming, a claim thought to be supported by Hilbert~\cite{Grunbaum-1978}. 
In fact, Reinhardt (and Hilbert?) were mistaken: no such proof is possible and Heesch provided the first counterexample in 1935~\cite{Heesch-1935} (see Figure~\ref{fig:anisohedrality-ex}).

\begin{figure}[ht]
\centering
\includegraphics[scale=1.0]{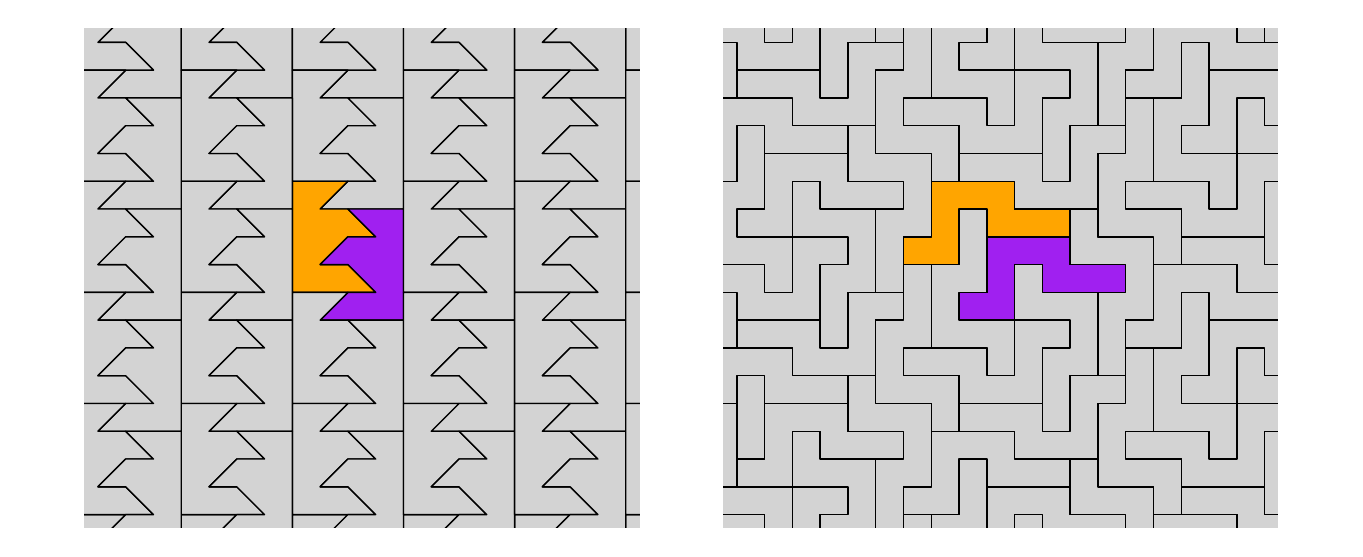}
\caption{The anisohedral polygon of Heesch~\cite{Heesch-1935} and an anisohedral polyomino of Rhoads~\cite{Rhoads-2005}. There is no symmetry of either tiling mapping one colored tile to the other.}
\label{fig:anisohedrality-ex}
\end{figure}

In the 1963, Heesch and Kienzle~\cite{Heesch-1963} provided the first complete classification of isohedral tilings.
This classification was given as nine \emph{boundary criteria}: conditions on a polygon's boundary that, if satisfied, imply an isohedral tiling and together form a necessary condition for isohedral polygons. 
Each boundary criterion describes a factorization of the boundary into a specific number of intervals with given properties, e.g., an interval is rotationally symmetric or two intervals are translations of each other.
Special cases of this classification have been rediscovered since, including the criterion of Beauquier and Nivat~\cite{Beauquier-1991} and Conway's criterion, attributed to John H. Conway by Gardner~\cite{Gardner-1975,Schattschneider-1980} (see Figure~\ref{fig:boundary-criteria}).


\begin{figure}[hb!]
\centering
\includegraphics[scale=1.0]{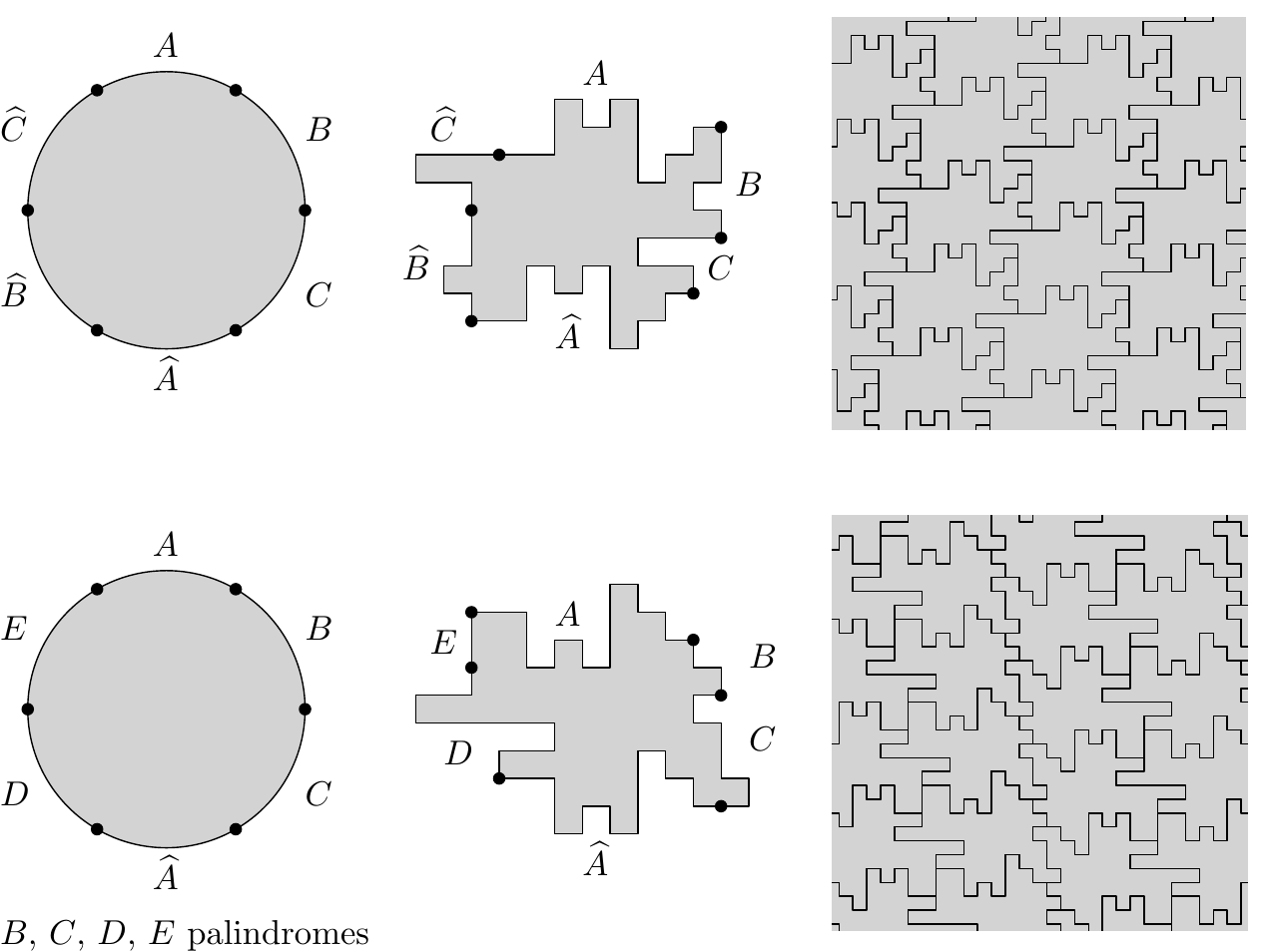}
\caption{Two of seven boundary criteria characterizations of isohedral tilings.
These criteria were given by Beauquier and Nivat~\cite{Beauquier-1991} (top) and John H. Conway~\cite{Gardner-1975} (bottom).
Precise definitions are given in Section~\ref{sec:definitions}.}
\label{fig:boundary-criteria}
\end{figure}

While a complete classification of isohedral tilings exists, many problems in tiling classification and algorithmics remain open. 
For instance, complete classifications of pentagons that tile the plane were claimed as early as 1968~\cite{Kershner-1968}, but additional pentagons have been discovered as recently as 2015~\cite{Mann-2015}.
The existence of an algorithm for deciding if a polyomino tiles the plane is a longstanding open problem~\cite{GoodmanStrauss-2000,GoodmanStrauss-2010}, as is the existence of a polygon that tiles only without symmetry~\cite{Socolar-2011}.

One of the most successful lines of work in tiling algorithmics was initiated by Wijshoff and van Leeuwen~\cite{Wijshoff-1984}, who considered tiling the plane using \emph{translated} copies of a polyomino (isohedrally or otherwise).
They proved that deciding whether a polyomino admits such a tiling is possible in polynomial time.
Their algorithm was subsequently improved by Beauquier and Nivat~\cite{Beauquier-1991}, who gave a simple boundary criterion for polyominoes that admit such a tiling.
Subsequent application of more advanced algorithmic techniques led to a sequence of improved algorithms by Gambini and Vuillon~\cite{Gambini-2007}, Proven\c{c}al~\cite{Provencal-2008}, Brlek, Proven\c{c}al, and F\'{e}dou~\cite{Brlek-2009a}, and the second author~\cite{Winslow-2015}, who gave an optimal $O(n)$-time algorithm, where $n$ is the number of edges on the polyomino's boundary.

The boundary criterion of Beauquier and Nivat matches one of the criteria of Heesch and Kienzle, implying that this problem is a special case of deciding if a polyomino is isohedral.
The general problem of isohedrality was proved decidable in 1999 by Keating and Vince~\cite{Keating-1999}, who gave a matrix-based algorithm running in $O(n^{18})$ time. 
Their algorithm does not make use of boundary criteria, which we note yields a straightforward $O(n^6)$-time algorithm.

Here we give a $O(n\log^2{n})$-time algorithm for deciding if a polyomino is isohedral.
The algorithm uses the original boundary characterization of Heesch and Kienzle~\cite{Heesch-1963} to decompose the problem into seven subproblems, each of recognizing whether a polyomino's boundary admits a factorization with a specific form.
Structural and algorithmic results on a variety of word problems are used, extending the approach of~\cite{Winslow-2015} to factorizations of six additional forms.
The algorithm also finds a witness tiling and is easily extended to other classes of lattice shapes, e.g. polyhexes and polyiamonds.

%



\section{Definitions}
\label{sec:definitions}

Although the main result of the paper concerns highly geometric tilings, the proof is entirely described using \emph{words}, also called \emph{strings}.
We use the term ``word'' for consistency with terminology in previous work on tilings of polyominoes.

\textbf{Polyomino and Tiling.}
A \emph{polyomino} is a simply connected polygon whose edges are unit length and parallel to one of two perpendicular lines.
Let $\mathscr{T} = \{T_1, T_2, \dots \}$ be an infinite set of finite simply connected closed sets of $\mathbb{R}^2$.
Provided the elements of $\mathscr{T}$ have pairwise disjoint interiors and cover the Euclidean plane, then $\mathscr{T}$ is a \emph{tiling} and the elements of $\mathscr{T}$ are called \emph{tiles}.
Provided every $T_i \in \mathscr{T}$ is congruent to a common shape $T$, then $\mathscr{T}$ is \emph{monohedral} and $T$ is the \emph{prototile} of $\mathscr{T}$.
In this case, $T$ is said to \emph{have} a tiling.
A monohedral tiling is also \emph{isohedral} provided, for every pair of elements $T_i, T_j \in \mathscr{T}$, there exists a symmetry of $\mathscr{T}$ that maps $T_i$ to $T_j$.
Otherwise the tiling is \emph{anisohedral}.

\textbf{Letter.}
A \emph{letter} is a symbol $x \in \Sigma = \{\Up, \Down, \Left, \Right\}$ representing the directions up, down, left and right.
The \emph{$\D{\Theta}$-ro\underline{t}ation} of a letter $x$, written $\Rot{\Theta}{x}$, is defined as the letter obtained by rotating $x$ counterclockwise by $\D{\Theta}$, e.g., $\Rot{270}{\Up} = \Right$.
A special case of $\D{\Theta}$-rotations is the \emph{complement} of a letter, written $\Comp{x}$ and defined as $\Comp{x} = \Rot{180}{x}$.
The \emph{$\D{\Theta}$-re\underline{f}lection} of a letter $x$, written $\Refl{\Theta}{x}$, is defined as the letter obtained by reflecting $x$ across a line with angle $\D{\Theta}\in \{\D{-45}, \D{0}, \D{45}, \D{90}\}$ counterclockwise from the x-axis, e.g., $\Refl{45}{\Up} = \Right$. 

\textbf{Word and Boundary Word.}
A \emph{word} is a sequence of letters and the \emph{length} of a word $W$, denoted $|W|$, is the number of letters in $W$.
For an integer $i \in \{1, 2, \dots, |W|\}$, $W[i]$ refers to the $i$th letter of $W$ and $W[-i]$ refers to the $i$th from the last letter of $W$.
The notation $W^i$ denotes the repetition of a word $i$ times.
In this work two kinds of words are used: non-circular words and circular words (defining the boundaries of polyominoes).
A word is \emph{non-circular} if it has a first letter, and \emph{circular} otherwise.
For a circular word $W$, an arbitrary but fixed assignment of the letter $W[1]$ may be used, resulting in a non-circular \emph{shift} of $W$.
The \emph{boundary word} of a polyomino $P$, denoted $\Bou{P}$, is the circular word of letters corresponding to the sequence of directions traveled along cell edges during a clockwise traversal of the polyomino's boundary (see Figure~\ref{fig:boundary-word-examples}).

\begin{figure}[ht]
\centering
\includegraphics[scale=1.0]{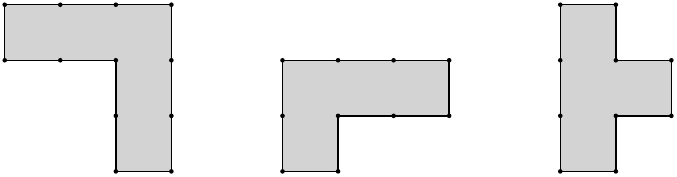}
\caption{Polyominoes with (circular) boundary words $\Up \Right^3 \Down^3 \Left \Up^2 \Left^2$, $\Up^2 \Right^3 \Down \Left^2 \Down \Left$, and $\Up^3 (\Right \Down)^2 \Left \Down \Left$ (from left to right).}
\label{fig:boundary-word-examples}
\end{figure}

\textbf{Rotation, complement, and reflection.}
The rotation (or complement, reflection) of a word $W$, written $\Rot{\Theta}{W}$ (or $\Comp{W}$, $\Refl{\Theta}{W}$), is the word obtained by replacing each letter in $W$ with its rotation (or complement, reflection).
The \emph{reverse} of a word $W$, written $\Rev{W}$, are the letters of $W$ in reverse order.
The \emph{backtrack} of a word $W$ is denoted $\Back{W}$ and defined as $\Back{W} = \Comp{\Rev{W}}$.

\textbf{Factor.}
A \emph{factor of $W$} is a contiguous sequence $X$ of letters in $W$, written $X \Factor W$.
For integers $1 \leq i, j \leq |W|$ with $i \leq j$, $W[i..j]$ denotes the factor of $W$ from $W[i]$ to $W[j]$, inclusive.
A factor $X$ \emph{starts} or \emph{ends} at $W[i]$ if $W[i]$ is the first or last letter of $X$, respectively.
Two factors $X, Y \Factor W$ may refer the same letters of $W$ or merely have the same letters in common.
In the former case, $X$ and $Y$ are \emph{equal}, written $X = Y$, while in the latter, $X$ and $Y$ are \emph{congruent}, written $X \equiv Y$.
For instance, if $W = \Up \Up \Up \Left \Right \Up \Up \Up$ then $W[1..3] \equiv W[6..8]$.
A \emph{factorization} of $W$ is a partition of $W$ into consecutive factors $F_1$ through $F_k$, written $W = F_1 F_2 \dots F_k$.

\textbf{Prefix, suffix, affix, middle, and center.}
A factor $X \Factor W$ is a \emph{prefix} if $X$ starts at $W[1]$, written $X \Prefix W$.
Similarly, $X \Factor W$ is a \emph{suffix} if $X$ ends at $W[-1]$, written $X \Suffix W$.
A factor $X \Factor W$ that is either a prefix or suffix is an \emph{affix}, written $X \Affix W$.
A factor $X \Factor W$ that is not an affix is a \emph{middle}, written $X \Middle W$. 
The factor $X \Factor W$ such that $W = UXV$, $|U|=|V|$, and $|X| \in \{1, 2\}$ is the \emph{center of $W$}.
Similar definitions for words are defined equivalently, e.g., a word is a \emph{prefix} of another word provided it is congruent to a
 prefix factor of that word.

\textbf{Period, composite, and primitive.}
A word \emph{$X$ is a period of $W$} provided $W$ is congruent to a prefix of $X^k$ for some $k \geq 0$ (introduced by~\cite{Knuth-1977}).
Alternatively, $X$ is a prefix of $W$ and $W[i] = W[i+|X|]$ for all $1 \leq i \leq |W|-|X|$.
A word $X$ is \emph{composite} provided there exists a subword $Y$ such that $X = Y^k$ for some $k \geq 2$, and otherwise is \emph{primitive}.

\textbf{$\Theta$-drome, square, and mirror.} 
A word $X$ is a \emph{$\Theta$-drome} provided $X = Y \Rot{\Theta+180}{\Rev{Y}}$, e.g., a \emph{palindrome} is a $180$-drome.
Such a factor is \emph{admissible} provided $W = XU$ with $U[-1] \neq \Rot{\Theta+180}{U}[1]$.\footnote{In this work, several types of factors and pairs of factors have restricted ``admissible'' versions. 
Intuitively, admissible versions are maximal in a natural sense for each type.
For instance, a $\Theta$-drome is admissible if it is the longest $\Theta$-drome with its center.}

A word $X$ is a \emph{square} provided $X = Y^2$ for some $Y$, and a \emph{reflect square} provided $X = Y \Refl{\Theta}{Y}$ for some $\Theta$.
Such a factor is \emph{admissible} provided $W = XU$ with $Y[-1] \neq \Refl{\Theta}{U[-1]}$ and $U[1] \neq \Refl{\Theta}{Y}[1]$. 

A factor $X \Factor W$ is a \emph{mirror}, written $X \Mirror W$, provided $W = XUYV$ with $Y \equiv \Back{X}$ and $|U|=|V|$.
Such a factor is \emph{admissible} provided $U[1] \neq \Comp{U[-1]}$, $V[1] \neq \Comp{V[-1]}$

\textbf{Gapped mirror and (reflect) square.}
A pair of disjoint factors $X, Y \Factor W$ is a \emph{gapped mirror} provided $X \equiv \Back{Y}$.
Such a pair $X, Y$ is \emph{admissible} provided $W = XUYV$ with $U[1] \neq \Comp{U[-1]}$, $V[1] \neq \Comp{V[-1]}$.

A pair of disjoint factors $X, Y \Factor W$ is a \emph{gapped (reflect) square} provided $X \equiv Y$ ($X \equiv \Refl{\Theta}{Y}$).
Such a pair is \emph{admissible} provided $W = XUYV$ with $U[1] \neq V[1]$ ($U[1] \neq \Refl{\Theta}{V[1]}$) and $U[-1] \neq V[-1]$ ($U[-1] \neq \Refl{\Theta}{V[-1]}$).

\section{Proof Overview}
\label{sec:overview}

The remainder of the paper is dedicated to proving the following theorem:

\begin{theorem}
Let $P$ be a polyomino with $|\Bou{P}| = n$.
It can be decided in $O(n\log^2{n})$ time if $P$ has an isohedral tiling (of the plane).
\end{theorem}

Here we survey some of the ideas involved in the proof.
The proof starts with a list of the boundary word factorization forms that together characterize the polyominoes capable of tiling isohedrally.
These are found in the bordered subregion\footnote{A translation of the caption of Table 10: ``The strong border contains~9 major types, from which the others can be thought of as emerging by shrinking lines or line pairs.''} of Table 10 of~\cite{Heesch-1963}\footnote{Reproduced on page 326 of~\cite{Schattschneider-1990}.} excluding the two types of isohedral tilings that use $\D{60}$ and $\D{120}$ rotations of the shape. 
The factorizations can be cross-verified using the incidence and adjacency symbols of a more detailed classification of isohedral tiling types of Gr\"{u}nbaum and Shephard~\cite{Grunbaum-1977}, and correspond to the isohedral types IH~1,~4,~28,~2,~3,~5, and~6 in this classification.
The factorization forms are:
\begin{itemize}
\item Translation: $ABC\Back{A}\Back{B}\Back{C}$.
\item Half-turn: $W = ABC\Back{A}DE$ with $B$, $C$, $D$, $E$ palindromes.
\item Quarter-turn: $W = ABC$ with $A$ a palindrome and $B$, $C$ $90$-dromes.
\item Type-1 reflection: $W = A B \Refl{\Theta}{B} \Back{A} C \Refl{\Phi}{C}$ for some $\Theta$, $\Phi$.
\item Type-2 reflection: $W = A B C \Back{A} \Refl{\Theta}{C} \Refl{\Theta}{B}$.
\item Type-1 half-turn-reflection: $W = A B C \Back{A} D \Refl{\Theta}{D}$ with $B$, $C$ palindromes.
\item Type-2 half-turn-reflection: $W = A B C D \Refl{\Theta}{B} \Refl{\Phi}{D}$ with $A$, $C$ palindromes and $\D{\Theta}-\D{\Phi} = \pm \D{90}$.
\end{itemize}

Winslow~\cite{Winslow-2015} gave a $O(n)$-time algorithm for deciding if a boundary word of length $n$ has a translation factorization.
Sections~\ref{sec:HLF} through~\ref{sec:T2H} give $O(n\log^2{n})$-time or faster algorithms for deciding if a boundary word of length $n$ has a factorization of each remaining form.
Brlek, Koskas, and Proven\c{c}al~\cite{Brlek-2009b} provide a $O(n)$-time algorithm for deciding if a given circular word is the boundary word of a polyomino, and we assume for the remainder of the paper that the input is guaranteed to be the boundary word of a polyomino and thus simple. 
This assumption of simplicity is used to prove that factors and pairs of factors in a factorization are \emph{admissible}: maximal in a natural sense for each factor (pair) type.
E.g., for half-turn factorizations: 

\vspace*{10pt}
\noindent
\textbf{Lemma~\ref{lem:HLF-half-turn-factors-admissible}.}
Let $P$ be a polyomino and $\Bou{P} = A B C \Back{A} D E$ with $B, C, D, E$ palindromes.
Then the gapped mirror pair $A$, $\Back{A}$ and palindromes $B$, $C$, $D$, $E$ are admissible.
\vspace*{10pt}

For various factors and factor pairs, this implies there are $O(n)$ or $O(n\log{n})$ candidate factors for these elements of a factorization, and they can be computed in similar time.
Note that this alone is not sufficient to solve these problems in quasi-linear time.
Without additional structural results, attempting to combine a quasi-linear number of factors into a factorization with one of the~6 forms is at least as hard as searching for cycles of fixed length between~3 and~6 in a graph with $n$ vertices and $O(n)$ edges, only known to admit a $O(n^{1.67})$-time algorithm~\cite{Alon-1997}.
Additional structural results must be used, such as the following for half-turn factorizations: 

\vspace*{10pt}
\noindent
\textbf{Lemma~\ref{lem:HLF-prefix-palin-factor}.}
\emph{The prefix palindrome factorization of a word $W$ has the form $W = X_1^{r_1} X_2^{r_2} \dots X_m^{r_m} Q$ with:
\begin{itemize}
\item Every $X_i$ primitive.
\item For all $3 \leq i \leq m$, $\sum_{j=1}^{i-2}{|X_j^{r_j}|} \leq |X_i|$ and thus $m = O(\log{|W|})$.
\end{itemize}}
\vspace*{10pt}

Such results allow more efficient ``batch processing'' of factors to achieve quasi-linear running time.
It can also be seen by cursory examination that each algorithm returns affirmatively only once witness factorization is found.
Witness factorizations define the set of boundary intervals shared by pairs of neighboring tiles in an isohedral tiling, thus the algorithm can also return a witness isohedral tiling if desired.


\section{Half-Turn Factorizations}
\label{sec:HLF}

\begin{definition}
\label{defn:half-turn-factorization}
A \emph{half-turn factorization} of a boundary word $W$ has the form $W = ABC\Back{A}DE$ with $B$, $C$, $D$, $E$ palindromes.
\end{definition}

\subsection{Prefix palindrome factorizations}
\label{sec:prefix-palindrome-factorizations}

\begin{definition}
Let $W$ be a word.
A factorization $W = F_1 F_2 \dots F_{n+1}$ is a \emph{prefix palindrome factorization} of $W$ provided that the set of prefix palindromes of $W$ is $\{F_1 F_2 \dots F_i : 1 \leq i \leq n \}$.
\end{definition}

\begin{lemma}
\label{lem:HLF-in-between-palin}
Let $W = PX$ with $P$, $W$ palindromes and $0 < |P| < |W|$.
Then $W$ has a period of length $|X|$. 
Furthermore, if $X$ is composite, then $W$ has a prefix palindrome longer than $P$.
\end{lemma}

\begin{proof}
Since $W = PX$ and $P$, $W$ are palindromes, $W = \Rev{W} = \Rev{PX} = \Rev{X} \Rev{P} = \Rev{X} P$.
So $P$ is a prefix of $\Rev{X} P$ and so $\Rev{X}$ is a period of $P$ and of $W = \Rev{X} P$.
Since $\Rev{X}$ is a period of $W$ and $|\Rev{X}| < |W|$, there exist words $Y$, $Z$ such that $\Rev{X} = ZY$, $W = (ZY)^p Z$, and $P = (ZY)^{p-1} Z$ for some $p \geq 1$.
So $X = YZ$ and since $W$ is a palindrome, $Y$ and $Z$ are palindromes.

If $X$ is composite, then $X = G^k$ for some $k \geq 2$.
So there exist words $G_1$, $G_2$ such that $G = G_1 G_2$, $Y = (G_1 G_2)^i G_1$, $Z = G_2 (G_1 G_2)^{k-i-1}$.
Since $Y$ and $Z$ are palindromes, $G_1$ and $G_2$ are palindromes.

Now we construct a prefix palindrome of $W$, called $Q$, that is longer than $P$.
Without loss of generality, assume $|G_2| > 0$. 
Then there are two possibilities for the values of $|G_2|$ and $|Z|$:
\begin{enumerate}
\item $|G_2| < |Z|$ and we let $Q = \Rev{X}^p G_2$.
\item $|G_2| = |Z|$ and we let $Q = PG_2$.
\end{enumerate}

In the first case, $Q = \Rev{X}^p G_2 = (G_1 G_2)^p G_2 = (G_2 G_1)^{kp} G_2$, so $Q$ is a palindrome.
Also, $W = (ZY)^p Z$ and $G_2$ is a prefix of $Z$, so $Q$ is a prefix of $W$ and $|P| = |W| - |X| = |W| - |ZY| = p|ZY| - |Y| < |Q| = p|ZY| + |G_2| < p|ZY| + |Z| = |W|$.

In the second case, $Y = G_2^i$ and $Z = G_2^{k-i}$.
So $Q = P G_2 = (ZY)^{p-1} Z G_2 = G_2^{kp-k} G_2^{k-i} G_2 = G_2^{kp-i+1}$ and $Q$ is a palindrome.  
Also, $W = PX$ and $G_2$ is a prefix of $X$, so $Q$ is a prefix of $W$ and $|P| < |Q| = |P| + |G_2| < |P| + |X| = |W|$.
\end{proof}

The following is a well-known result; see Chapter~2 of Crochemore and Rytter~\cite{Crochemore-1994}.

\begin{lemma}[Fine and Wilf's theorem~\cite{Fine-1965}]
\label{lem:HLF-strong-periodicity}
Let $W$ be a word with periods of length $p$ and $q$.
If $p + q \leq |W|$, then $W$ also has a period of length ${\rm gcd}(p, q)$. 
\end{lemma}

\begin{lemma}
\label{lem:HLF-palin-size-steps}
Let $P_1, P_2, \dots, P_m$ be the set of prefix palindromes of a word with $0 < |P_1| < |P_2| < \dots < |P_m|$.
Then for any $1 \leq i \leq m-2$, either $|P_{i+1}| - |P_i| = |P_{i+2}| - |P_{i+1}|$ or $|P_i| + |P_{i+1}| < |P_{i+2}|$.
\end{lemma}

\begin{proof}
Let $P_{i+1} = P_i X_i$ and $P_{i+2} = P_i X_i X_{i+1}$.

Since there are no prefix palindromes of length between $|P_i|$ and $|P_{i+1}|$ or $|P_{i+1}|$ and $|P_{i+2}|$, Lemma~\ref{lem:HLF-in-between-palin} implies $X_i$ and $X_{i+1}$ are primitive and $P_{i+1}$ and $P_{i+2}$ have periods of length $|X_i|$ and $|X_{i+1}|$, respectively.
Since $P_{i+1}$ is a prefix of $P_{i+2}$, it also has a period of length $|X_{i+1}|$.

The lemma permits $|X_i| = |X_{i+1}|$, so assume $|X_i| \neq |X_{i+1}|$.
If $|X_{i+1}| \leq |P_i|$, then $P_{i+1}$ has periods of length $|X_i|$ and $|X_{i+1}|$ with $|X_{i+1}| + |X_i| \leq |P_{i+1}|$.
Then by Lemma~\ref{lem:HLF-strong-periodicity}, $P_{i+1}$ has a period of length ${\rm gcd}(|X_i|, |X_{i+1}|)$.
This length must be at least $|X_i|$ and $|X_{i+1}|$, otherwise $X_i$ or $X_{i+1}$ is not primitive. 
So $|X_i| = |X_{i+1}|$.
Otherwise, $|X_{i+1}| > |P_i|$ and so $|P_i| + |P_{i+1}| < |X_{i+1}| + |P_{i+1}| = |P_{i+2}|$.
\end{proof}

The next lemma is a strengthening of similar prior results by Apostolico, Breslauer, and Galil~\cite{Apostolico-1995}, I et al.~\cite{I-2014}, and Matsubara et al.~\cite{Matsubara-2009}.

\begin{lemma}[Prefix Palindrome Factorization Lemma] 
\label{lem:HLF-prefix-palin-factor}
The prefix palindrome factorization of a word $W$ has the form $W = X_1^{r_1} X_2^{r_2} \dots X_m^{r_m} Q$ with:
\begin{itemize}
\item Every $X_i$ primitive.
\item For all $3 \leq i \leq m$, $\sum_{j=1}^{i-2}{|X_j^{r_j}|} \leq |X_i|$ and thus $m = O(\log{|W|})$.
\end{itemize}
\end{lemma}

\begin{proof}
We give a constructive proof.
Let $P_1, P_2, \dots, P_n$ be the set of prefix palindromes of $W$ with $|P_1| < |P_2| < \dots < |P_n|$.

Let $W_i$ be the word such that $P_{i+1} = P_i W_i$ and let $Q$ be the word such that $W = P_n Q$.
So $W$ has a prefix palindrome factorization $W_1 W_2 \dots W_n Q$.
By Lemma~\ref{lem:HLF-in-between-palin}, every $W_i$ is primitive.
Moreover, by Lemma~\ref{lem:HLF-palin-size-steps} either $|W_i| = |W_{i+1}|$ or $|P_i| < |W_{i+1}|$ for every $1 \leq i \leq n-2$.

Suppose $|W_i| = |W_{i+1}|$.
By Lemma~\ref{lem:HLF-palin-size-steps}, $P_i$ and $P_{i+1}$ have a common period and thus $W_i = W_{i+1}$.
More generally, if $|W_i| = |W_{i+1}| = \dots = |W_{i+c}|$, then $W_i = W_{i+1} = \dots = W_{i+c}$.
If $|P_i| < |W_{i+1}|$, then $|W_{i+1}| > |P_i| = \sum_{j=1}^{i-1}|W_j|$.
So the factorization $W_1 W_2 \dots W_n Q$ can be rewritten as $X_1^{r_1} X_2^{r_2} \dots X_m^{r_m} Q$ with the property that $|X_i| \geq \sum_{j=1}^{i-2}{|X_j|^{r_j}}$.
So for all $i \geq 4$, $2|X_{i-3}| < |X_i|$ and thus $m = O(\log{|W|})$.
\end{proof}

Such a factorization can be stored using $O(\log{|W|})$ space by simply storing $|X_i|$ and $r_i$ for each $i$.
Additional observations can be used prove that $|W|$ prefix palindrome factorizations of the suffixes of a word $W$ can be computed in optimal time:

\begin{lemma}
\label{lem:HLF-all-prefix-facts-fast}
The prefix palindrome factorizations of all shifts of a circular word $W$ can be computed in $O(|W|\log{|W|})$ total time. 
\end{lemma}

\begin{proof}
Lemma 9 of~\cite{I-2014} states that the prefix palindrome factorization of a non-circular word $xY$ can be computed in $O(\log{|Y|})$ time given the factorization of $Y$.
Thus the factorizations of non-circular word $WW$ can be enumerated in $O(|W|\log{|W|})$ time, beginning with $Y = W[-1]$.
Every shift of word $W$ is a subword of the non-circular word $WW$, and the computed factorizations can be trimmed in $O(\log{|W|})$-time per factorization to be the factorizations of shifts of $W$.
\end{proof}

Identical results, including a suffix palindrome factorization lemma, clearly hold for suffix palindromes as well.

\subsection{Algorithm}

The main idea is to iterate over all pairs of adjacent letters and guess the form of the palindromes $D$ and $E$ in both directions from that location.
Specifically, guess what repeated factor $X_i^{r_i}$ terminates them in their prefix and suffix palindrome factorizations. 
Then try to complete the factorization using Lemma~\ref{lem:HLF-half-turn-subroutine}, which decides if it is possible to rewrite a given portion of the boundary as $L^b A B C \Back{A} R^c$ with $B$, $C$ palindromes and $b$, $c$ in some range. 
The results leading up to Lemma~\ref{lem:HLF-half-turn-subroutine} provide the necessary structure to achieve this goal.
In particular, Lemma~\ref{lem:HLF-dp-finish} shows how to decompose a word into two palindromes, and Lemmas~\ref{lem:HLF-palindrome-pump} through~\ref{lem:HLF-small-candidate-rep-set} yield fast detection of a factorization of the form $BCR^k$ with $B$, $C$ palindromes.

\begin{lemma}
\label{lem:HLF-palindrome-pump}
Let $W$ be a word with subwords $L$, $R$ such that $W = L R^r$ and $R \not \Suffix L$.
Let $P_1$, $P_2$ be palindromes such that $W = P_1 P_2 R^k$ with $|L| \leq |P_1|$.
Then there exists a palindrome $P_2'$ and integer $k'$ such that $W = P_1 P_2' R^{k'}$ with $|P_2'| < |R|$.
\end{lemma}

\begin{proof}
Since $|L| \leq |P_1|$, $P_2$ is a a suffix of $R^i$ for some minimal $i$.
If $i \leq 1$, then either $|P_2| = |R|$ (and $|P_2'| = 0$, $k' = k+1$) or $|P_2| < |R|$ and the claim is satisfied.
If $i \geq 2$, then there exist words $Y$, $Z$ with $|Y| > 0$ such that $R = YZ$ and $P_2 = Z(YZ)^{i-1}$.

Let $P_2' = Z$ and $k' = i-1+k$.
So $W = P_1 P_2 R^k = P_1 P_2' (Y P_2')^{i-1} R^k = P_1 P_2' R^{k'}$.
Since $|YZ| = |Y P_2'| = |R|$ and $|Y| > 0$, it follows that $|P_2'| < |R|$.
Since $P_2 = Z (YZ)^{i-2}Y Z$, $Z = P_2'$ is a palindrome.
\end{proof}

\begin{lemma}[Lemma~C4 of~\cite{Galil-1978}]
\label{lem:HLF-extremal-double-palindromes}
If a word $X_1 X_2 = Y_1 Y_2 = Z_1 Z_2$ with $X_2$, $Y_1$, $Y_2$, $Z_1$ palindromes.
Then $X_1$ and $Z_2$ are palindromes.
\end{lemma}

\begin{lemma}
\label{lem:HLF-small-candidate-rep-set}
Let $R$ be a primitive word and let $W = LR^r$.
There is a set of integers $H$ with $|H| = O(\log{|W|})$ such that $W = P_1 P_2 R^k$ if and only if it does so with $|LR^h| \leq |P_1| \leq |LR^{h+1}|$ for some $h \in H$.
Moreover, given $|R|$ and the prefix palindrome factorization of $W$, $H$ can be computed in $O(\log{|W|})$ time.
\end{lemma}

\begin{proof}
We may assume $|P_2| < |R|$ by Lemma~\ref{lem:HLF-palindrome-pump}.
Consider the prefix palindrome factorization of $W$ as described in Lemma~\ref{lem:HLF-prefix-palin-factor}.
Any solution $P_1$ ends with one of the repeating subwords $X_i$ of the factorization.
There are three cases: $|X_i| < |R|$, $|X_i| > |R|$, and $|X_i| = |R|$.

\textbf{Case 1: $\bm{|X_i| < |R|}$.}
We claim that if $|X_i| < |R|$, then $X_i^{r_i}$ overlaps $R^r$ in at most two repetitions of $R$ (and there are at most three values of $h$).
Assume, for the sake of contradiction, that $R^2$ is a subword of $|X_i^{r_i}|$.
Then $R^2$ is a word of length at least $|X_i| + |R|$ with periods of length $|X_i|$ and $|R|$.
So by Lemma~\ref{lem:HLF-strong-periodicity}, $R$ has a period of length ${\rm gcd}(|X_i|, |R|) \leq |X_i| < |R|$, a contradiction.

\textbf{Case 2: $\bm{|X_i| > |R|}$.}
We claim that if $|X_i| > |R|$, then $X_i$ cannot repeat in $R^r$ (and there are at most two values of $h$).
Assume, for the sake of contradiction, that $X_i^2$ is a subword of $R^r$.
Then by Lemma~\ref{lem:HLF-strong-periodicity}, $X_i$ has a period of length ${\rm gcd}(|X_i|, |R|) \leq |R| < |X_i|$.
So by Lemma~\ref{lem:HLF-in-between-palin}, the factorization given was not a prefix palindrome factorization, a contradiction.

\textbf{Case 3: $\bm{|X_i| = |R|}$.}
We claim that if $|X_i| = |R|$, then $h = 0$ suffices.
Suppose $|LR^h| \leq |P_1| \leq |LR^{h+1}|$ for some $h \geq 1$. 
By Lemma~\ref{lem:HLF-in-between-palin}, $P_1$ has a period of length $|X_i| = |R|$.
Let $Y$, $Z$ be words such that $YZ$ is a period of $P_1$ and $|R| - |Z| = |P_2|$.
So $P_1 = (YZ)^p Y$ for some $p \geq 1$ and $Y$, $Z$ are palindromes.

Since $LR^{h+1} = (YZ)^p Y P_2$ and $|YZYP_2| = 2|R|$, $YZ = YP_2 = R$. 
So $LR = (YZ)^{p-h} Y P_2 = P_1' P_2$, where $P_1' = (YZ)^{p-h} Y$ and thus is a palindrome.
So there exists a $P_1'$ with $|LR^0| \leq |P_1'| \leq |LR^1|$. 

\textbf{Computing $H$.}
The value of $h$ in case 3 is always $0$.
For case~1, use the values of $h$ such that $X_i^{r_i}$ contains the last letter of $LR^h$ or $LR^{h+1}$.
For case~2, use the values of $h$ such that the prefix palindrome ending at the unique repetition of $X_i$ has length between $|LR^h|$ and $|LR^{h+1}|$.
\end{proof}

\begin{lemma}
\label{lem:HLF-double-palin-side-fast}
Let $R$ be a primitive word and let $W = LR^r$.
Assume that $|R|$ and the prefix and suffix palindrome factorizations of $W$ are given.
Then it can be decided in $O(\log{|W|})$ time if $W = P_1 P_2 R^k$ with $P_1$, $P_2$ palindromes and $|L| \leq |P_1|$.
\end{lemma}

\begin{proof}
We may assume $|P_2| < |R|$ by Lemma~\ref{lem:HLF-palindrome-pump}.
First, use Lemma~\ref{lem:HLF-small-candidate-rep-set} to compute a $O(\log{|W|})$-sized candidate set of integers $H$ such that a solution exists if and only if $LR^{h+1} = P_1 P_2$ with $|LR^{h}| \leq |P_1| \leq |LR^{h+1}$ for some $h \in H$.
By Lemma~\ref{lem:HLF-extremal-double-palindromes}, it suffices to check for such solutions with at least one of the following types of palindromes:
\begin{itemize}
\item The longest prefix palindrome of $LR^{h+1}$ with length at least $|LR^h|$.
\item The longest suffix palindrome of $LR^{h+1}$ with length less than $|R|$.
\end{itemize}

Compute the longest prefix palindromes of $LR^{h+1}$ for all values of $h$ in $O(\log{|W|})$ total time using a two-finger scan of (1) the prefix palindrome factorization of $W$ and (2) the values of $h$.
Use a second two-finger scan of (1) these prefix palindromes and (2) the suffix palindrome factorization of the last $|R|$ letters of the suffix palindrome factorization of $LR^r$ to search for a solution $P_1, P_2$.

The longest suffix palindrome of $LR^{h+1}$ (with length less than $R$) is invariant for $h$ and can be computed in $O(\log{|W|})$ time, using the last $|R|$ letters of the suffix palindrome factorization of $LR^r$.
Call the length of this palindrome $\lambda$.
Use a scan of the prefix palindrome factorization to determine if a prefix palindrome of $W$ has length $|LR^h|-\lambda$ for some value of $h$.  
\end{proof}

Next, we develop a second result that is combined with the previous lemma to obtain Lemma~\ref{lem:HLF-half-turn-subroutine}.

\begin{lemma}
\label{lem:HLF-period-common-prefix}
Let $L$ and $R$ be words such that $L \not \Prefix R$.
Let $A_i$ be the longest common prefix of $L^{l-i} R$ and a word $U$.
Let $k = l - \lceil |A_0|/|L| \rceil$.
Then:
\begin{itemize}
\item For all $i$ with $0 \leq i \leq k$, $A_i = A_0$ and $|A_i| < |L^{l-k}|$.
\item For all $i$ with $k+2 \leq i \leq l$, $|A_i| - |L^{l-i}| = |A_{k+2}| -  |L^{l-(k+2)}|$.
\end{itemize}  
\end{lemma}

\begin{proof}
Let $k$ be the maximum $k$ such that $|A_k| < |L^{l-k}|$.
We show that this value of $k$ has the desired properties, including that $k = l - \lceil |A_0|/|L| \rceil$.

\textbf{Property~1.}
Let $0 \leq i \leq k$.
Since $A_i$ is a prefix of $L^{l-i} R$ and $|A_i| < |L^{l-i}|$, $A_i$ is the longest common prefix of $L^l$ and $U$.
This is true for all choices of $A_i$, and thus all $A_i$ are equal.

\textbf{Property~2.}
By definition, $|A_{k+1}| \geq |L^{l-(k+1)}|$ and so $L^{l-(k+1)}$ is a prefix of $U$.
Since $L \not \Prefix R$, the length of $A_{k+2}$, the longest common prefix of $L^{l-(k+2)} R$ and $U$, must be less than $|L^{l-(k+1)}|$.
So $|A_{k+2}| < |L^{l-(k+1)}|$ and $L$ is a period of $A_{k+2}$.

Let $R_1$, $R_2$ be words such that $R = R_1 R_2$ and $A_{k+2} = L^{l-(k+2)} R_1$.
Then $|R_1| < |L|$ and so $R_1$ is the longest common prefix of $R$ and $L$.

So for all $i \geq k+2$, the longest common prefix of $L^{l-i} R$ and $L^{l-(k+1)}$ is $L^{l-i} R_1$.
Moreover, since $|L^{l-i} R_1| < |L^{l-(k+1)}|$ and $L^{l-(k+1)}$ is a prefix of $U$, the longest common prefix of $L^{l-i} R$ and $U$ is also $L^{l-i} R_1$.

\textbf{Property 3.}
Finally, we prove that $k = l - \lceil |A_0|/|L| \rceil$.
Since $|A_0| = |A_i| < |L^{l-i}|$ for all $0 \leq i \leq k$, it follows that $|A_0| < |L^{l-k}| = |L|(l-k)$.
Then by algebra, $k < l - |A_0|/|L|$.

Since $A_{k+1} \geq |L^{l-(k+1)}|$, it must be that $U$ has a prefix $L^{l-(k+1)}$.
So $|A_0| = |A_k| \geq |L^{l-(k-1)}| = |L|(l-k-1)$.
Then by algebra $k + 1 \geq l - |A_0|/|L|$.
So $k < l - |A_0|/|L| \leq k+1$.
\end{proof}

\begin{lemma}
\label{lem:HLF-dp-finish}
Let $W$ be a word and $l_1$, $l_2$ integers.
Assume the prefix and suffix palindrome factorizations of $W$ are given.
Then it can be decided in $O(\log{|W|})$ time if there exist palindromes $P_1$, $P_2$ such that $W = P_1 P_2$ with $|P_1| \geq l_1$, $|P_2| \geq l_2$. 
\end{lemma}

\begin{proof}
By Lemma~\ref{lem:HLF-extremal-double-palindromes}, such a pair of palindromes exist if and only if there exists such a pair such that either $P_1$ is the longest prefix palindrome of $W$ with $|P_1| \leq |W|-l_2$ or $P_2$ is the longest suffix palindrome of $W$ with $|P_2| \leq |W|-l_1$. 
Scan each factorization in $O(\log{|W|})$ time to find these specific palindromes, and then scan the opposite factorizations for a second palindrome to complete $W$. 
\end{proof}

The following result comes from a trivial modification of Theorem 9.1.1 of~\cite{Gusfield-1997} to allow for circular words, namely giving the concatenation of two copies of a corresponding non-circular word as input, and returning $\infty$ if the output has length more than ${\rm lcm}(|X|, |Y|)$.

\begin{lemma}[Theorem 9.1.1 of~\cite{Gusfield-1997}]
\label{lem:HLF-longest-common-extension}
Two circular words $X$, $Y$ can be preprocessed in $O(|X| + |Y|)$ time to support the following queries in $O(1)$-time:
what is the longest common factor of $X$ and $Y$ starting at $X[i]$ and $Y[j]$?
\end{lemma}

\begin{lemma}
\label{lem:HLF-half-turn-subroutine}
Let $W$ be a circular word.
Let $W = L^l Z$ and $W = Y R^r$ such that $L \not \Prefix Z$, $R \not \Suffix Y$, and $L$, $R$ are primitive.
Assume that the prefix and suffix palindrome factorizations of every shift of $W$ are given.
It can be decided in $O((l+r)\log{|W|})$ time if there exist positive integers $b$, $c$ such that $W = L^b A P_1 P_2 \Back{A} R^c$ with $A$, $\Back{A}$ admissible and $P_1$, $P_2$ palindromes.
\end{lemma}

\begin{proof}
The approach is to iteratively search for a solution for each value of $b$, carrying out the same algorithm on each value.
Then performing an identical, symmetric iteration through the values of $c$.
First assume $b$ is a fixed value and $c$ is not.
Let $A_i$ be the longest word such that $L^b A_i$ and $\Back{A_i} R^i$ are a prefix and suffix of $W$, respectively.
Lemma~\ref{lem:HLF-period-common-prefix} implies that there exists an integer $k$ such that:
\begin{itemize}
\item For all $i$ with $0 \leq i \leq k$, $|L^b A_i|$ is fixed and $|\Back{A_i} R^i| = |A_0| + |R|i$.
\item For all $i$ with $k+2 \leq i \leq l$, $|\Back{A_i} R^i|$ is fixed and $|L^b A_i| = |L^b| + |A_{k+2}| - |R|(i-(k+2))$.
\item $k$ can be computed in $O(1)$ time assuming a data structure allowing $O(1)$ time longest common prefix queries for suffixes of $W$ and $\Back{W}$ is given. 
\end{itemize}
In other words, $k$ is an efficiently-computable integer that partitions the values of $i$ into three parts: one with a single value ($i = k+1$) and two others where either $|L^b A_i|$ or $|\Back{A_i} R^i|$ is fixed and the other is a linear set.
Handle the case of $i = k+1$ individually by using Lemma~\ref{lem:HLF-dp-finish} to check if the word between $L^b A_{k+1}$ and $\Back{A_{k+1}} R^{k+1}$ has a factorization into two palindromes.
Next, check that all $A_i$, $\Back{A_i}$ except $i = 0$ are admissible by verifying $L^b[-1] \neq \Comp{R^i[1]}$.
Also handle $i = 0$ individually, including checking admissibility. 
Lemma~\ref{lem:HLF-double-palin-side-fast} is used to handle the remaining two cases in $O(\log{|W|})$ time each.

\textbf{Case 1: $\bm{0 \leq i \leq k}$.}
In this case, $|L^b A_i|$ is fixed and $|\Back{A_i} R^i| = |A_0| + |R|i$.
If $k-1 < 3$, then handle all three cases individually in $O(\log{|W|})$ time using Lemma~\ref{lem:HLF-dp-finish}. 
Otherwise handle only $i = k$ similarly.

By Lemma~\ref{lem:HLF-period-common-prefix}, $A_i = A_0$ for all $0 \leq i \leq k-1$ and $|A_0| < |R^{r-k}|$.
So $\Back{A_0} R^i$ for all $0 \leq i \leq k-1$ and $R^r$ are suffixes of $W$.
Also, for all $i \leq k-1$, $|\Back{A_0} R^i| \leq |\Back{A_0} R^{k-1}| \leq |\Back{A_0} R^k| - |R| \leq |R^r| - |R|$.
So for some $R'$ with $|R'| = |R|$, $\{ \Back{A_0} R^i : 0 \leq i \leq k-1\} = \{ (R')^i \Back{A_0} : 0 \leq i \leq k-1 \}$.
Let $L'$ be such that $W = L^b A_i L' (R')^k \Back{A_0}$.
Then a solution factorization exists if and only if there exist palindromes $P_1 P_2 = L' (R')^{k-i}$ for some $0 \leq i \leq k-1$.

First, search for solutions with $|P_1| \geq |L'|$.
We first prove that $R'$ is primitive, allowing Lemma~\ref{lem:HLF-double-palin-side-fast} to be invoked.
Suppose, for the sake of contradiction, that $R'$ has a period of length $p$ with $p < |R'|$.
So $(R')^2$ has periods of length $|R'|$ and $p$ such that $|R'| + p < |(R')^2|$ and so by Lemma~\ref{lem:HLF-strong-periodicity} has a period of length ${\rm gcd}(|R'|, p) < |R'|$.
Then since $(R')^2$ contains $R$ as a subword, $R$ also has a period of length $p$ and thus is not primitive, a contradiction.

For solutions with $|P_1| < |L'|$, Lemma~\ref{lem:HLF-extremal-double-palindromes} implies that it suffices to check for solutions with the longest possible $P_1$ (longest possible $P_2$ is handled when performing the symmetric iteration over values of $c$).
To do so, scan the prefix palindrome factorization starting at $L'[|P_1| + 1]$ for a palindrome of length $|L'| - |P_1| + |(R')^{k-i}|$ with $0 \leq i \leq k-1$.

\textbf{Case 2: $\bm{k+2 \leq i \leq l}$.}
In this case, $|\Back{A_i} R^i|$ is fixed and $|L^b A_i| = |L^b| + |A_{k+2}| - |R|(i-(k+2))$.
Then there exists a word $R'$ such that $W = L^b (\Back{R})^{r-i} A_r R' \Back{A_i} R^i$ for all $k+2 \leq i \leq l$.
Let $L'$ be the suffix of $\Back{R} A_r$ of length $|R|$.
Then $W = L^b A_r (L')^{r-(k+2)} R' \Back{A_r} R^r$.
So there exists a pair of palindromes $P_1$, $P_2$ with $W = L^b A_i P_1 P_2 \Back{A_i} R^i$ for some $k+2 \leq i \leq l$ if and only if $(L')^{r-i} R' = P_1 P_2$ for some $k+2 \leq i \leq l$.
This situation is identical to that encountered in the previous case -- handle in the same way.

\textbf{Handling overlap.}
The description of the algorithm so far has ignored the possibility that $|L^b A_i| + |\Back{A_i} R^i| > |W|$, i.e., that $L^b A_i$ and $\Back{A_i} R^i$ ``overlap''.
For the case of $0 \leq i \leq k$, this occurs when $|L^b A_0| + |\Back{A_0} R^i| > |W|$.
Restricting the values of $i$ to satisfy $0 \leq i \leq {\rm min}(k, \lfloor (|W|-b|L|-2|A_0|)/|R| \rfloor)$ ensures that $W$ can be decomposed as claimed.
For the case of $k+2 \leq i \leq l$, this occurs when $|L^b R^{r-i} A_r| + |\Back{A_r} R^r| > |W|$.
Restricting the values of $i$ to satisfy ${\rm max}(\lceil (2|A_r| + 2r|R| + b|L| - |W|)/|R| \rceil, k+2) \leq i \leq l$ ensures that $W$ can be decomposed as claimed. 
Check the individually handled cases, namely $i = k, k+1$, for overlap individually.

\textbf{Running time.}
The running time of this algorithm is $O((l+r)\log{|W|})$, since the amount of time spent for each value of $b$ and $c$ is $O(\log{|W|})$ to handle individual values of $i$ and $O(\log{|W|})$ to handle each large case by Lemma~\ref{lem:HLF-double-palin-side-fast}.
However, this assumes a data structure enabling $O(1)$ time common prefix queries on $W$ and $\Back{W}$.
Compute such a data structure in $O(|W|)$ time using Lemma~\ref{lem:HLF-longest-common-extension}.
Since $\Omega(|W|)$ time must be spent to decide if a boundary word has a half-turn factorization, such a computation has no additional asymptotic cost.
\end{proof}

\begin{lemma}
\label{lem:HLF-half-turn-factors-admissible}
Let $P$ be a polyomino and $\Bou{P} = A B C \Back{A} D E$ with $B, C, D, E$ palindromes.
Then the gapped mirror pair $A, \Back{A}$ and palindromes $B$, $C$, $D$, $E$ are admissible.
\end{lemma}

\begin{proof}
\textbf{$A, \Back{A}$ is admissible.}
It cannot be that $|B| = |C| = 0$, since then $ABC\Back{A} = A\Back{A}$ is non-simple.
If $BC[1] = \Comp{BC}[-1]$, then $|B|, |C| > 0$ and thus $B[-1] = BC[1] = \Comp{BC}[-1] = \Comp{C}[1]$ and $BC$ is non-simple.
So $BC[1] \neq \Comp{BC}[-1]$ and by symmetry, $DE[1] \neq \Comp{DE}[-1]$.
So $A, \Back{A}$ are admissible.

\textbf{$B$, $C$, $D$, $E$ are admissible.}
Consider the pairs of non-equal consecutive letters in $W$.
These pairs come from sets $\mathcal{R} = \{\Left\Up, \Up\Right, \Right\Down, \Down\Left\}$ and $\mathcal{L} = \{\Up\Left, \Left\Down, \Down\Right, \Right\Up\}$, and Proposition~6 of~\cite{Daurat-2005} states that the number of pairs from $\mathcal{R}$ is four more than the number from $\mathcal{L}$. 
Also, any palindrome contains an equal number of consecutive letter pairs from $\mathcal{L}$ and $\mathcal{R}$.

If $|A| = 0$, then $W$ has factorization $W = BCDE$ with the four consecutive-letter pairs from $\mathcal{R}$ not contained in any factor, i.e., for each factor $X \in \{B, C, D, E\}$, $W = XY$ with $Y[-1]X[1], X[-1]Y[1] \in \mathcal{R}$.
Since $X$ is a palindrome, $X[-1]Y[1] = X[1]Y[1] \in \mathcal{R}$ and so $Y[1] \neq Y[-1]$.
Thus $X$ is admissible.

If $|A| > 0$, then $|BC| > 0$, since otherwise $A[-1]\Back{A}[1] = A[-1]\Comp{A[-1]}$ is a subword and $W$ is non-simple.
Without loss of generality, $|B| > 0$.
If $|C| = 0$, then $W = BY$ with $Y[1] = \Back{A}[1] = \Comp{A}[-1] \neq A[-1] = Y[-1]$ and $B$ is admissible.
If $|C| > 0$, then $C[1] = C[-1] \neq \Comp{\Back{A}}[1] = A[-1]$.
So $W = BY$ with $Y[-1] = A[-1] \neq C[1] = Y[1]$ and so $B$ is admissible.
By symmetry, it is also the case that $C$, $D$ and $E$ are also admissible. 
\end{proof}

\begin{theorem}
\label{thm:HLF-algorithm}
Let $P$ be a polyomino with $|\Bou{P}| = n$.
It can be decided in $O(n\log^2{n})$ time if $\Bou{P}$ has a half-turn factorization.
\end{theorem}

\begin{proof}
First, compute the prefix palindrome factorizations of each shift of $W$ by computing and truncating the prefix palindrome factorizations of the  $|W|$ longest suffixes of $W$ using Lemma~\ref{lem:HLF-all-prefix-facts-fast}. 
Similarly compute the suffix palindrome factorization of every shift of $W$.
By Lemma~\ref{lem:HLF-all-prefix-facts-fast}, this takes $O(|W|\log{|W|})$ total time. 

Next, compute the admissible factors (including zero-length factors), i.e., palindromes maximal about their center, by computing and truncating the maximal palindromes of $WW$ output by Manacher's $O(|W|)$-time algorithm~\cite{Manacher-1975}.
Each admissible factor $F$ is contained in a prefix palindrome factorization as $X_1^{r_1} X_2^{r_2} \dots X_i^j$ with either $1 \leq i \leq m$ and $0 \leq j < r_i$, or $i = m$ and $j = r_m$ if $F$ is the longest prefix palindrome of the word.
If $j < r_i$, call $X_i^{r_i}$ the \emph{terminator} of $F$, otherwise call $Q^1$ the terminator of $F$.
Either all or none of the prefix palindromes with a given terminator are admissible.
Similar definitions and observations apply to suffix palindrome factorizations.

For each admissible factor $W$, mark the two terminators (one prefix, one suffix) of the factor. 
Locating the prefix and suffix terminators for each of the $2|W|$ admissible factors takes $O(|W|\log{|W|})$ total time. 

Without loss of generality, every solution half-turn factorization has $|E| > 0$.
Search for half-turn factorizations $ABC\Back{A}DE$ by iterating over possible first letters of $E$.
By Lemma~\ref{lem:HLF-half-turn-factors-admissible}, only solutions with admissible $D$ and $E$ must be considered.
This corresponds to a palindromes $D$ and $E$ starting and ending with a marked terminators $X_s^{r_s}$ and $X_p^{r_p}$, respectively.
For each such terminator pair $X_s^{r_s}, X_p^{r_p}$, use Lemma~\ref{lem:HLF-half-turn-subroutine} with $L = X_p$, $l = r_p$, $R = \Rev{X_s}$, $r = r_s$ to check for a partial factorization $ABC\Back{A}$ to complete the factorization along with $D$ and $E$ with the marked terminator pair. 
By Lemma~\ref{lem:HLF-prefix-palin-factor}, $X_s$ and $X_p$ are primitive and by definition of palindrome factorizations, $l$ and $r$ are maximal.

Checking for a partial factorization using Lemma~\ref{lem:HLF-half-turn-subroutine}, each pair $X_s^{r_s}, X_p^{r_p}$ takes $O((r_s + r_p)\log{|W|})$ time, $O(\log{|W|})$ time for each admissible factor involved.
Moreover, each admissible factor is involved in $O(\log{|W|})$ pairs of terminators: $O(\log{|W|})$ prefix (suffix) terminators when $E$ ($D$). 
So $O(\log^2{|W|})$ total time is spent per admissible factor and in total the the algorithm takes $O(|W|\log^2{|W|})$ time.
\end{proof}


\section{Quarter-Turn Factorizations}
\label{sec:QRT}

\begin{definition}
A \emph{quarter-turn factorization} of a boundary word $W$ has the form $W = ABC$ with $A$ a palindrome and $B$, $C$ $90$-dromes.
\end{definition}

\subsection{Long $90$-dromes}

\begin{lemma}
\label{lem:QRT-no-ddrome-periodic}
Let $W$ be a word with a period of length $p$, and $X$ a $90$-drome subword of $W$.
Then $|X| \leq p$.
\end{lemma}

\begin{proof}
Suppose, for the sake of contradiction and without loss of generality, that $|X| = p+2$.
Since $W$ has a period of length $p$, $X$ has a period of length $p$.
So $X[1] = X[p+1] = X[-2]$ and $X[2] = X[2+p] = X[-1]$.
Since $X$ is a $90$-drome, $X[1] = \Rot{90}{X[-1]}$ and $X[2] = \Rot{90}{X[-2]}$.
So $X[1] = \Rot{90}{X[-1]} = \Rot{90}{X[2]} = \Rot{180}{X[-2]} = \Rot{180}{X[1]}$, a contradiction.
\end{proof}

\begin{lemma}
\label{lem:QRT-third-ddrome-bound}
Let $W$ be a word. 
Let $P \Prefix W$, $S \Suffix W$ be distinct $90$-dromes with $|P|, |S| \geq 2/3|W|$.
Then any other $90$-drome factor of $W$ with center not shared with $P$ or $S$ has length less than $2|W|-(|P|+|S|)$.
\end{lemma}

\begin{proof}
Let $P = P_L P_R$ and $S = S_L S_R$ with $|P_L| = |P_R|$ and $|S_L| = |S_R|$.
So there exists $A \Factor W$ such that $W = P_L A S_R$ (see Figure~\ref{fig:third-ddrome-bound}).
Since $|P_L| + |A| + |S_R| = |W|$, $|A| \leq |W|/3$.
Then because $|P_R|, |S_L| \leq |W|/3$, it follows that $A \Factor P_R, S_L$. 

\begin{figure}[ht]
\centering
\includegraphics[scale=1.0]{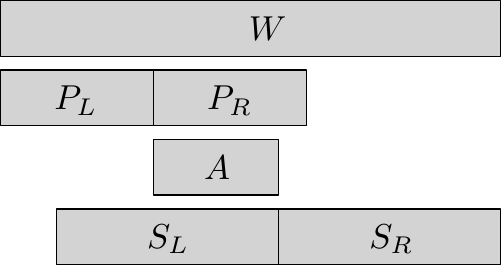}
\caption{The words used in the proof of Lemma~\ref{lem:QRT-third-ddrome-bound}.}
\label{fig:third-ddrome-bound}
\end{figure}

Let $1 \leq l, r \leq |W|$ such that $r = l + 2|A|$, $l < |P_L|$, and $(l+r)/2 \neq (|P|-1)/2$ (the center of $W[l..r]$ is not the center of $P$).
The approach is to prove $W[l] \neq \Rot{90}{W[r]}$ and thus any $90$-drome factor with center not shared with $P$, or $S$ by symmetry, has length at most $r-l = 2|A| = 2|W|-(|P|+|S|) < 2|W|-2/3(|P|+|S|)$. 
Since for any pair $1 \leq l, r \leq |W|$ such that $r-l = 2|A|$, either $l \leq |P_L|$ or $|P_L A| < r$, there are two cases: $l \leq |P_L|-|A|$ and $|P_L|-|A| < l \leq |P_L|$.

\textbf{Case 1: $\bm{l \leq |P_L|-|A|}$.}
\begin{equation*}
\begin{split}
W[l] &= P_L[l] \\ 
&= \Rot{90}{P_R[|P_R|-l]} \\
&= \Rot{90}{S_R[|P_R|-l-|A|]} \\
&= S_L[|S_L|-(|P_R|-l-|A|)] \\
&= P_LA[|P_LA|-(|P_R|-l-|A|)] \\
&= P_LA[l + 2|A|] \\
&= W[r] \\
\end{split}
\end{equation*}

\textbf{Case 2: $\bm{|P_L|-|A| < l \leq |P_L|}$.}
\begin{equation*}
\begin{split}
W[l] &= P_L[l] \\
&= \Rot{90}{P_R[|P_R|-l]} \\
&= \Rot{90}{A[|P_R|-l]} \\
&= \Rot{90}{S_L[|S_L|-|A|+(|P_R|-l)]} \\
&= \Rot{90}{S_L[|S_L|+(|P_R|-|A|-l)]} \\
&= \Rot{180}{S_R[-(|P_R|-|A|-l)]} \\
&= \Rot{180}{W[|P_L|+|A|-(|P_R|-|A|-l)]} \\
&= \Rot{180}{W[l+2|A|]}
\end{split}
\end{equation*}

\end{proof}

\begin{lemma}
\label{lem:QRT-prefix-ddrome-exp-growth}
Let $P, Q, W$ be $90$-dromes with $P, Q \Prefix W$ and $|P| < |Q| < |W|$.
Then $|P| < 2/3|W|$.
\end{lemma}

\begin{proof}
If $|Q| \leq 2/3|W|$, then clearly the lemma holds.
Otherwise, Lemma~\ref{lem:QRT-third-ddrome-bound} implies that $|P| < 2|W|-2/3(|P|+|W|) = 4/3|W|-2/3|P| < 4/3|W|-2/3|W| = 2/3|W|$.
\end{proof}

\begin{lemma}
\label{lem:QRT-long-ddrome-constant}
Let $W$ be a boundary word.
There exist $O(1)$ admissible $90$-drome factors of $W$ with length at least $|W|/3$.
\end{lemma}

\begin{proof}
Consider the set $\mathscr{I}$ of all such factors with centers contained in a factor of $W$ with length at most $|W|/16$.
Let $G$ be the shortest factor of $W$ that contains all factors in $\mathscr{I}$.
Either there exist two distinct factors $P, S \in \mathscr{I}$ such that $P \Prefix G$ and $S \Suffix G$ or $G \in \mathscr{I}$.

First, assume the former.
Since the centers of $P$ and $S$ lie in a common factor of length at most $|W|/16$, $P \Prefix G$, and $S \Suffix G$, $|P|+|W|/8 \geq |S|$.
So $|G| \leq |P|/2 + |W|/16 + |S|/2 \leq |P| + |W|/16 + (|P|+|W|/8)/2 \leq |P|+|W|/8 \leq 3/2|P|$ and thus $2/3|G| \leq |P|$.
By symmetry, also $2/3|G| \leq |S|$.
Moreover, since every element of $\mathscr{I}$ is admissible, no pair of elements share the same center.
Then by Lemma~\ref{lem:QRT-third-ddrome-bound}, any factor $F \in \mathscr{I}$ with $F \neq P, S$ has length less than $2|G|-(|P|+|S|) \leq |P|+|W|/4+|S|-(|P|+|S|) = |W|/4 < |W|/3$.
So $\mathscr{I} = \{P, S\}$.

Next, assume the latter.
Let $G'$ be the shortest factor with the same center as $G'$ that contains all elements of $\mathscr{I}$ except $G$.
Since $G'$ is not admissible, $G' \not \in \mathscr{I}$.
Without loss of generality, there exists $P' \in \mathscr{I}$ such that $P' \Prefix G'$.
Using the same argument as before, $2/3|G'| \leq |P'|$.
Also, clearly $2/3|G'| \leq |G'|$ and since $P'$ is admissible and $G'$ is not, $P' \neq G'$.
Then any factor $F \in \mathscr{I}$ with $F \neq P', G', G$ has length less than $2|G'|-(|P'|+|G'|) = |G'|-|P'| \leq |W|/4$.
So $\mathscr{I} = \{P', G\}$.

Thus the maximum number of admissible $90$-drome factors of $W$ with length at least $|W|/3$ is at most $17|\mathscr{I}| = 34$, obtained by partitioning all such factors into at most~17 groups with centers contained in common factors of length at most $|W|/16$.
\end{proof}

\subsection{Long palindromes}

\begin{lemma}
\label{lem:QRT-presuf-palin-period}
Let $W$ be a word.
Let $P \Prefix W$, $S \Suffix W$ with $|P|, |S| \geq 2/3|W|$, $P \neq S$, and $P$, $S$ palindromes.
Then $W$ has a period of length $2|W|-(|P|+|S|)$.
\end{lemma}

\begin{proof}
Let $P = P_L P_R$ and $S = S_L S_R$ with $|L_P| = |R_P|$ and $|L_S| = |R_S|$.
The proof proceeds identically to that of Lemma~\ref{lem:QRT-third-ddrome-bound}, except that since elements equidistant from the center of a palindrome are equal, rather than at a $\D{90}$ rotation as in a $90$-drome, $W[i] = W[i+2|A|]$ for all $i$ with either $1 \leq i \leq |P_L|-|A|$ (case~1) or $|P_L|-|A| < i \leq |P_L|$ (case~2). 
Thus, by symmetry, $W[i] = W[i+2|A|]$ for all $i$ with $1 \leq i \leq |W|-2|A|$.
\end{proof}

\begin{lemma}
\label{lem:QRT-middle-palins-nonadmissible}
Let $W$ be a word.
Let $P \Prefix W$, $S \Suffix W$ be distinct palindromes with $|P|, |S| \geq 2/3|W|$.
Then any admissible middle palindrome factor of $W$ has length at most $2|W|-(|P|+|S|)$.
\end{lemma}

\begin{proof}
By Lemma~\ref{lem:QRT-presuf-palin-period}, $W$ has a period of length $p = 2|W|-(|P|+|S|)$.
Let $W[l-1..r+1]$ be a middle palindrome factor of $W$ with length more than $p$.
So $(r-1)-(l+1)+1 > p$ and thus $r-l \geq p+2$.

Consider $W[r-pi]$, $W[l+pi]$ for integers $i$ such that $1 \leq l+pi, r-pi \leq |W|$.
Since $W$ has a period of length $p$, $W[l+pi]$ is equal for all such values of $i$, as is $W[r-pi]$.

Let $f(i) = (r-pi)-(l+pi)$.
Then $f(0) = r-l \geq p+2$ and $f(i+1) = f(i)-2p$, so there exists an integer $x > 0$ such that $|f(x)| \leq p$.
Also, $W[l+px]$ and $W[r-px]$ are equidistant from the center of a palindrome of length at least $p+1$.
So $W[l+px] = W[r-px]$.
Combining this with the previous observation, $W[l] = W[l+px] = W[r-px] = W[r]$.
So $W[l+1..r-1]$ is not admissible.
\end{proof}

\begin{lemma}
\label{lem:QRT-long-palindrome-ends-constant}
Let $W$ be a word.
There exists a $O(1)$-sized set $\mathscr{F}$ of factors of $W$ such that every admissible palindrome factor with length at least $|W|/3$ is an affix factor of an element of $\mathscr{F}$.
Moreover, $\mathscr{F}$ can be computed in $O(|W|)$ time.
\end{lemma}

\begin{proof}
\textbf{Three factors.}
Let $P, Q, S$ be distinct admissible palindrome factors of $W$ with length at least $|W|/3$ and centers contained in a factor of $W$ with length at most $|W|/16$.

Let $G \Factor W$ be the shortest factor such that $P, Q, S \Factor G$.
Now it is proved that if $G \neq P, Q, S$, then $P, Q, S \Affix G$.

Without loss of generality, suppose $P \Prefix G$, $S \Suffix G$.
Since the centers of $P$ and $S$ lie in a common factor of length at most $|W|/16$ and $S \neq G$, $|P| + |W|/8 \geq |S|$.
So $|G| \leq |P|/2 + |W|/16 + |S|/2 \leq |P|/2 + |W|/16 + (|P|+|W|/8)/2 = |P|+|W|/8 \leq 3/2|P|$ and thus $2/3|G| \leq |P|$.
By symmetry, also $2/3|G| \leq |S|$.
Then by Lemma~\ref{lem:QRT-middle-palins-nonadmissible}, $|Q| \leq 2|G| - (|P|+|S|) \leq |P| + |W|/8 + |S| - (|P|+|S|) = |W|/8$.
So $Q \not \Middle G$ and thus $Q \Affix G$.

\textbf{More than three factors.}
Consider a set $\mathscr{I} = \{F_1, F_2, \dots, F_m\}$ of at least three admissible factors of $W$ of length at least $|W|/3$ such that the centers of the factors are contained in a common factor of $W$ of length $|W|/16$.
Now it is proved that every element of $\mathscr{I}$ is an affix factor of one of the two factors.

Let $G \Factor W$ be the shortest factor such that $F \Factor G$ for every $F \in \mathscr{I}$.
Either there exist distinct $P, S \in \mathscr{I}$ with $P \Prefix G$, $S \Suffix G$, or that $G \in \mathscr{I}$ and every $F \in \mathscr{I}$ besides $G$ has $F \Middle G$.

In the first case, the previous claim regarding three factors implies $F \Affix G$ for any $F \in \mathscr{I}$ with $F \neq P, S$.
Also $P, S \Affix G$.
So every factor in $\mathscr{I}$ is an affix factor of $G$.

In the second case, let $G' \Factor G$ be the shortest factor with the same center as $G$ such that every factor in $\mathscr{I}$ excluding $G$ is a factor of $G'$.
Since $G'$ is not admissible, $G' \not \in \mathscr{I}$.
Without loss of generality, there exists $P \in \mathscr{I}$ such that $P \Prefix G'$.
Since $P \in \mathscr{I}$ and $G' \not \in \mathscr{I}$, $P \neq G'$.
Using the same argument as before, $2/3|G'| \leq |P|$.
Then by Lemma~\ref{lem:QRT-middle-palins-nonadmissible}, every middle factor of $G'$ in $\mathscr{I}$ has length at most $2|G'|-(|G'|+|P|) \leq |G'| - |P| \leq |W|/3$.
So every factor of $G'$ in $\mathscr{I}$ is an affix factor of $G'$.
Thus every factor in $\mathscr{I}$ is either $G$ or an affix factor of $G'$.

\textbf{All factors.}
Partition $W$ into 17 factors $I_1, I_2, \dots, I_{17}$ each of length at most $|W|/16$.
Let $\mathscr{I}_i$ be the set of factors with centers containing letters in $I_i$.
Then by the previous claim regarding more than three factors, there exists a set $\mathscr{F}_i$ ($G$ and possibly $G'$) such that every element of $\mathscr{I}_i$ is an affix factor of an element of $\mathscr{F}_i$ and $|\mathscr{F}_i| \leq 2$.
So every admissible $F \Mirror W$ with $|F| \geq |W|/3$ is an affix factor of an element of $\mathscr{F} = \bigcup_{i=1}^{17}{\mathscr{F}_i}$ and $|\mathscr{F}| \leq 2\cdot17$.

\textbf{Computing $\bm{\mathscr{F}}$.}
The proof as described nearly yields the algorithm.
Use Manacher's algorithm~\cite{Manacher-1975} to compute all $O(|W|)$ admissible palindrome factors of $W$ in $O(|W|)$ time.
Remove all factors shorter than $|W|/3$ and group the remainder into~17 groups $I_1, \dots, I_{17}$ as previously described.
For each group of palindromes, sort the indices of the first and last letters separately.
Use the first two values in the sorted lists, including ties, to compute and output $G$, and $G'$, if necessary.
\end{proof}

\subsection{Algorithm}


\begin{lemma}
\label{lem:QRT-factors-admissible}
Let $P$ be a polyomino and $\Bou{P} = ABC$ with $A$ a palindrome and $B$, $C$ $90$-dromes.
Then $A$, $B$, $C$ are admissible.
\end{lemma}

\begin{proof}
\textbf{$A$ is admissible.}
Without loss of generality, $|B| > 0$ and either $|C| = 0$ or $|C| > 0$.
If $|C| = 0$, then $Y[1] = B[1] = \Rot{90}{B[-1]} \neq \Rot{0}{Y[-1]}$.
If $|C| > 0$, then $Y[1] = B[1] = \Rot{90}{B[-1]} \neq \Rot{90}{\Comp{C[1]}} = \Rot{270}{(\Rot{90}{C[-1]})} = \Rot{0}{Y[-1]}$.

\textbf{$B$, $C$ are admissible.}
If $|A| = 0$, then $|B|, |C| > 0$.
So $C[1] = \Rot{90}{C[-1]} \neq \Rot{270}{C[-1]}$.
So $W = BC$ with $C[1] \neq \Rot{270}{C[-1]}$ and thus $B$ is admissible.
If $|A| > 0$, then $C[1] = \Rot{-90}{C[-1]} \neq \Rot{-90}{\Comp{A[1]}} = \Rot{270}{A[1]} = \Rot{270}{A[-1]}$.
So $W = BY$ with $Y[1] = C[1] \neq \Rot{270}{A[-1]}$ and thus $B$ is admissible.
Combining both cases, $B$ is always admissible.
By symmetry, $C$ is also always admissible.
\end{proof}

\begin{theorem}
Let $P$ be a polyomino with $|\Bou{P}| = n$.
It can be decided in $O(n)$ time if $\Bou{P}$ has a quarter-turn factorization.
\end{theorem}

\begin{proof}
\textbf{Computing admissible factors.}
Use Lemma~\ref{lem:HLF-longest-common-extension} to preprocess $W$ and $\Rev{W}$ in $O(|W|)$ time to allow $O(1)$-time computation of the longest palindrome for each center $W[i]$ or $W[i..i+1]$, i.e., the admissible palindrome factors of $W$.
Do the same with $W$ and $\Rot{90}{\Rev{W}}$ to compute the admissible $90$-dromes of $W$.
For each letter $W[i]$, construct a length-sorted list of the admissible $90$-dromes that start at $W[i]$; do the same for those that end at $W[i]$.
Also construct similar lists of the admissible palindromes.

\textbf{Long factors.}
By pigeonhole, any quarter-turn factorization has at least one \emph{long} factor of length at least $|W|/3$.
Lemma~\ref{lem:QRT-long-ddrome-constant} states that there are $O(1)$ long $90$-drome factors, while Lemma~\ref{lem:QRT-prefix-ddrome-exp-growth} implies that the long palindrome factors can be summarized by a $O(|W|)$-time-computable, $O(1)$-sized set of letters that contains either the first or last letter of every such factor.
Finally, Lemma~\ref{lem:QRT-prefix-ddrome-exp-growth} implies that for any letter, there are $O(\log{|W|})$ $90$-drome factors that start (or end) at the letter, and thus $O(\log^2{|W|})$ double $90$-drome factors that start (or end) at the letter.
The approach is to search for factorizations containing long palindrome or $90$-drome factors separately by guessing the $90$-drome factors of the factorization, given either a long $90$-drome factor or the location of the first or last letter of a long palindrome factor.
The correctness follows from observing that every factorization can be constructed in this way.

\textbf{Long $90$-drome factors.}
Scan through the $O(|W|)$ admissible $90$-drome factors and extract the $O(1)$ long factors.
For each such factor $W[i..j]$, use the previously computed lists to combine this factor with the admissible $90$-drome factors that end at $W[i-1]$ and start at $W[j+1]$, yielding all factorizations $W = A D_1 D_2$ such that $D_1$ or $D_2$ equals $W[i..j]$ and both are admissible $90$-dromes. 
Then check for each factorization if $A$ is a palindrome, again using the previously computed lists. 
There are $O(1)$ choices for $W[i..j]$, $O(\log{|W|})$ factorizations containing it, and checking if $A$ is a palindrome takes $O(\log{|W|})$ time.
So $O(\log^2{|W|})$ total time is spent.

\textbf{Large palindrome factors.}
Let $W[i]$ be a letter in the aforementioned set of letters summarizing the long palindrome factors.
Iterate through the double admissible $90$-dromes $D_1 D_2$ starting at $W[i+1]$ using the previous lists, including choices with $|D_2| = 0$.
Each choice induces a factorization $W = A D_1 D_2$; check if $A$ is a palindrome.
Repeat the same process, but for factors $D_1 D_2$ that end at $W[i-1]$.
There are $O(\log^2{|W|})$ choices for the double admissible $90$-domes, and checking if $A$ is a palindrome takes $O(\log{|W|})$ time.
So $O(\log^3{|W|})$ total time is spent.
\end{proof}



\section{Type-1 Reflection Factorizations}
\label{sec:T1R}

\begin{definition}
A \emph{type-1 reflection factorization} of a boundary word $W$ has the form $W = A B \Refl{\Theta}{B} \Back{A} C \Refl{\Phi}{C}$ for some $\Theta$, $\Phi$.
\end{definition}

\subsection{Enumerating square and reflect square factors}

This section starts by reviewing the $O(|W|\log{|W|})$-time algorithm of Main and Lorentz for enumerating all square factors of a word.
Their algorithm is then modified to output a set of $O(|W|\log{|W|})$ reflect square factors of a boundary word $W$ that contains all factors $B \Refl{\Theta}{B}$ and $C \Refl{\Phi}{C}$ found in type-1 reflection factorizations of $W$.

Let $W = UV$ with $|U| = |V| \pm 1$.
The algorithm by Main and Lorentz first enumerates all squares in $W$ that do not lie in $U$ or $V$, then recursively calls the algorithm on $U$ and $V$.
The key is a $O(|W|)$-time algorithm for enumerating the squares that do not lie in $U$ or $V$.
This algorithm relies on the following algorithmic and structural lemmas:

\begin{lemma}[\cite{Main-1984}]
\label{lem:T1R-main-lcp}
Let $A$, $B$ be words.
The longest common prefix of $A$ and $B[i..|B|]$ for all $1 \leq i \leq |B|$ can be computed in $O(|B|)$ total time.
\end{lemma}

\begin{lemma}[\cite{Main-1984}]
\label{lem:T1R-main-repetition}
Let $W = UV$.
Let $l, k$ be integers with $1 \leq l \leq |V|$ and $l \leq k \leq 2l$.
Let $s_i$ be the length of the longest common suffix of $U$ and $V[1..i]$, and $p_i$ be the length of the longest common prefix of $V$ and $V[i..|V|]$.
Then the factor $F$ with $|F|=2l$ and last letter at $V[k]$ is a square if and only if $2l - s_l \leq k \leq l + p_{l+1}$. 
\end{lemma}

\begin{proof}
\begin{figure}[ht]
\centering
\includegraphics[scale=1.0]{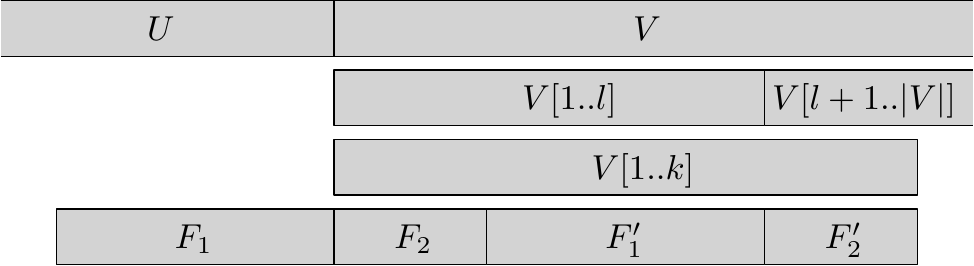}
\caption{The words used in the proof of Lemma~\ref{lem:T1R-main-repetition}.}
\label{fig:main-repetition-lemma}
\end{figure}
The factor $F$ is a square if and only if $F = F_1 F_2 F_1' F_2'$ such that $F_1 = F_1'$, $F_2 = F_2'$, $|F_1| = 2l-k$, and $|F_2| = k-l$ (see Figure~\ref{fig:main-repetition-lemma}).
In this case, $F_1$ is a suffix of $U$, $F_1'$ is a suffix of $UV[1..l]$, $F_2$ is a prefix of $V$, and $F_2'$ is a prefix of $V[l+1]$.
Thus the condition holds if and only if $s_l \geq 2l-k$ (ensuring $F_1 = F_1'$) and $p_{l+1} \geq k-l$ (ensuring $F_2 = F_2'$).
Combining these, $2l - s_l \leq k \leq l + p_{l+1}$.
\end{proof}

\begin{lemma}[\cite{Main-1984}]
\label{lem:T1R-main-enumeration}
Let $W$ be a word.
The square factors of $W$ can be enumerated in $O(|W|\log{|W|})$ time.
\end{lemma}

\begin{proof}
A square factor $F$ is \emph{new} provided $F \not \Factor U, V$.
A new square factor of $F$ is \emph{left} provided $F[1..|F|/2] \Factor U$ and \emph{right} otherwise.
Then Lemma~\ref{lem:T1R-main-repetition} implies the new right square factors of length $l$ are those ending at $V[k]$ for $k$ in the interval $[2l-s_l, l+p_{l+1}]$.
Use Lemma~\ref{lem:T1R-main-lcp} to compute $p_i$ and $s_i$ for all values of $i$ in $O(|V|)$ time.
Enumerate the new right square factors of each length $2l$ as a set defined by the interval $[2l-s_l, l+p_{l+1}]$ of last-letter indices in $O(|V|)$ total time.

By symmetry, computing the left square factors is possible in $O(|U|)$ time.
All new square factors are either left or right, so computing all such factors is possible in $O(|U|+|V|) = O(|W|)$ time.
Recursion on $U$ and $V$ suffices to compute all factors in total time $T(|W|) = 2T(|W|/2) + O(|W|) = O(|W|\log{|W|})$ time.
\end{proof} 

\begin{lemma}
\label{lem:T1R-square-enum}
Let $W$ be a word.
The reflect square factors of $W$ can be enumerated in $O(|W|\log{|W|})$ time.
\end{lemma}

\begin{proof}
For now, consider enumerating all reflect square factors $X \Refl{\Theta}{X}$ for fixed $\Theta$.
Recall the proof of Lemma~\ref{lem:T1R-main-enumeration}.
Redefine $s_i$ in the proof to be the longest common suffix of $\Refl{\Theta}{U}$ and $P[1..i]$, and similarly $p_i$ to be the longest common prefix of $\Refl{\Theta}{P}$ and $P[i..|P|]$.
Lemma~\ref{lem:T1R-main-lcp} still implies $p_i$, $s_i$ for all values of $i$ can be computed in $O(|V|)$ total time, and the remainder of the proof implies all reflect square factors $X \Refl{\Theta}{X}$ of $W$ can be enumerated in $O(|W|\log{|W|})$ time.
Perform the algorithm once for each choice of $\Theta$ to enumerate all reflect square factors in $O(|W|\log{|W|})$ time.
\end{proof}

Notice that the number of reflect square factors output is not necessarily $O(|W|\log{|W|})$, and in fact may be $\Omega(|W|^2)$, e.g., $W = (\Up \Right \Down \Right)^i$.
The running time is achieved by outputting factors in sets, each consisting of factors with the same length and contiguous last letters.
Next, we prove that only factors of singleton set can be reflect square factors $B \Refl{\Theta}{B}$ or $C \Refl{\Phi}{C}$ in type-1 reflection factorization.

\begin{lemma}
\label{lem:T1R-factors-admissible}
Let $P$ be a polyomino and $\Bou{P} = A B \Refl{\Theta}{B} \Back{A} C \Refl{\Phi}{C}$.
Then the gapped mirror pair $A$, $\Back{A}$ and reflect squares $B \Refl{\Theta}{B}$, $C \Refl{\Phi}{C}$ are admissible.
\end{lemma}

\begin{proof}
\textbf{$\bm{A, \Back{A}}$ is admissible.}
Consider $B[1]$ and $\Refl{\Theta}{B}[-1]$.
Since $B[-1]$ and $\Refl{\Theta}{B}[1]$ are adjacent, $B[-1] \neq \Refl{\Theta}{B}[1]$ and thus $\Refl{\Theta}{B}[-1] \neq B[1]$.
Similarly, $\Refl{\Phi}{C}[-1] \neq C[1]$, and thus $A, \Back{A}$ is admissible.

\textbf{$\bm{B \Refl{\Theta}{B}}$, $\bm{C \Refl{\Phi}{C}}$ are admissible.}
Consider the admissibility of $B \Refl{\Theta}{B}$; the other case follows by symmetry.
Either $|A| > 0$ or $|A| = 0$.
In the first case, suppose without loss of generality and for the sake contradiction, that $\Back{A}[1] = B[1] = x$.
Then $A[-1] = \Comp{\Back{A}[1]} = \Comp{x}$ and $\Bou{P}$ has a subword $x \Comp{x}$, a contradiction.

In the second card, consider the pairs of non-equal consecutive letters in $W$.
These pairs come from sets $\mathcal{R} = \{\Left\Up, \Up\Right, \Right\Down, \Down\Left\}$ and $\mathcal{L} = \{\Up\Left, \Left\Down, \Down\Right, \Right\Up\}$, and Proposition~6 of~\cite{Daurat-2005} states that the number of pairs from $\mathcal{R}$ is four more than the number from $\mathcal{L}$. 
Also, the total number of pairs from $\mathcal{L}$ and $\mathcal{R}$ in a pair of words $X$, $\Refl{\Theta}{X}$ is equal.
Thus every pair not in a factor of the factorization must be in $\mathcal{R}$.

Suppose, for the sake of contradiction, that $\Refl{\Theta}{B}[1] = x$ and $C[1] = \Refl{\Theta}{x}$.
Let $B[-1] = y$, so $\Refl{\Theta}{B}[-1] = \Refl{\Theta}{y}$.
Then $yx$ and $\Refl{\Theta}{yx}$ are consecutive pairs in $W$ not in any factor (they straddle factors $B, \Refl{\Theta}{B}$ and $\Refl{\Theta}{B}, C$, respectively).
Since $\Refl{\Theta}{y} \Refl{\Theta}{x}$ is the reflection of $yx$, these pairs cannot both in $\mathcal{R}$, a contradiction. 
\end{proof}

\begin{lemma}
\label{lem:T1R-small-fast-superset}
Let $W$ be a boundary word.
A $O(|W|\log{|W|})$-sized superset of all admissible factors $X = B \Refl{\Theta}{B}$ of $W$ can be enumerated in $O(|W|\log{|W|})$ time.
\end{lemma}

\begin{proof}
The algorithm is identical to that of Lemma~\ref{lem:T1R-square-enum}, except only return a set if $2l-s_l = l+p_{l+1}$, (see the proof of Lemma~\ref{lem:T1R-square-enum}) i.e., there is only one factor in the set. 
All that remains is to prove that if there is more than factor in a set, then no factor in the set can be in a type-1 reflection factorization.

Consider two factors $W[l..r] = Y \Refl{\Theta}{Y}$, $W[l+1..r+1] = Z \Refl{\Theta}{Z}$.
Let $x = W[r+1]$ and thus $\Refl{\Theta}{Y}[1] = Z[-1] = \Refl{\Theta}{x}$.
So $W[l..r]$ is not admissible. 
By symmetry, $W[l+1..r+1]$ is also not admissible.
Then by Lemma~\ref{lem:T1R-factors-admissible}, neither factor can be in a type-1 reflection factorization. 
So since each set of reflect square factors enumerated by the algorithm of Lemma~\ref{lem:T1R-square-enum} share a common length and contiguous set of last letters, the elements of the set can be factors of a type-1 reflection factorization only if the set is singleton.
\end{proof}

\subsection{Algorithm}

\begin{theorem}
\label{thm:T1R-algorithm}
Let $P$ be a polyomino with $|\Bou{P}| = n$.
It can be decided in $O(n\log{n})$ time if $\Bou{P}$ has a type-1 reflection factorization.
\end{theorem}

\begin{proof}
First, compute a set of $O(|W|\log{|W|})$ candidates for the admissible factors $B \Refl{\Theta}{B}$ in $O(|W|\log{|W|})$ time using Lemma~\ref{lem:T1R-small-fast-superset}.
Next, use Lemma~\ref{lem:HLF-longest-common-extension} to preprocess $W$, $\Back{W}$, and $\Refl{\Theta}{W}$ for $\Theta \in \{-\D{45}, \D{0}, \D{45}, \D{90}\}$ to enable the following queries in $O(1)$ time: ``For $Y, Z \in W, \Back{W}, \Refl{\Theta}{W}$, what is longest common prefix of $Y[i..|Y|]$ and $Z[j..|Z|]$?''.
For each candidate factor $B \Refl{\Theta}{B}$, use a query on $W$ and $\Back{W}$ to compute the induced $A$ and $\Back{A}$ factors, and then queries on $W$ and $\Refl{\Phi}{W}$ for all $\Phi$ to check if the remainder of $W$ is a reflection square $C \Refl{\Phi}{C}$.

By Lemma~\ref{lem:T1R-small-fast-superset}, computing the candidates for $B \Refl{\Theta}{B}$ takes $O(|W|\log{|W|})$ time.
For each candidate, five $O(1)$-time queries are used to determine if the candidate can be completed into a factorization.
So the algorithm takes $O(|W|\log{|W|})$ total time.
\end{proof}


\section{Type-2 Reflection Factorizations}
\label{sec:T2R}

\begin{definition}
A \emph{type-2 reflection factorization} of a boundary word $W$ has the form $W = A B C \Back{A} \Refl{\Theta}{C} \Refl{\Theta}{B}$ for some $\Theta$.
\end{definition}

\subsection{Gapped Reflect Squares}


\begin{lemma}
\label{lem:T2R-factors-admissible-1}
Let $P$ be a polyomino and $\Bou{P} = A B C \Back{A} \Refl{\Theta}{C} \Refl{\Theta}{B}$.
Then gapped reflect squares $B, \Refl{\Theta}{B}$ and $C, \Refl{\Theta}{C}$ are admissible. 
\end{lemma}

\begin{proof}
Consider the admissibility of $B, \Refl{\Theta}{B}$; the other case follows by symmetry.
Consider the pairs of non-equal consecutive letters in $W$.
These pairs come from sets $\mathcal{R} = \{\Left\Up, \Up\Right, \Right\Down, \Down\Left\}$ and $\mathcal{L} = \{\Up\Left, \Left\Down, \Down\Right, \Right\Up\}$, and Proposition~6 of~\cite{Daurat-2005} states that the number of pairs from $\mathcal{R}$ is four more than the number from $\mathcal{L}$.
Also, observe that the consecutive letter pairs in factors $A$, $\Back{A}$, $B$, and $\Refl{\Theta}{B}$ are divided evenly between $\mathcal{L}$ and $\mathcal{R}$.
Thus it cannot be that $|A|, |C| = 0$ and three cases remain.

\textbf{Case 1: $\bm{|A|, |C| > 0}$.}
Suppose, for the sake of contradiction and without loss of generality, that $A[1] = C[1] = x$.
So $\Refl{\Theta}{C}[1] = \Refl{\Theta}{x}$ and $\Back{A}[-1] = \Comp{A[1]} = \Comp{x}$.
Thus $W$ has a subword $\Comp{x} x$, a contradiction.

\textbf{Case 2: $\bm{|A| > 0, |C| = 0}$.}
Suppose, for the sake of contradiction and without loss of generality, that $\Back{A}[1] = \Refl{\Theta}{A[1]}$.
Let $x = B[-1]$ and $y = \Back{A}[1]$.
Thus $\Refl{\Theta}{B}[-1] = \Refl{\Theta}{x}$ and $A[1] = \Refl{\Theta}{y}$.
So $xy$ and $\Refl{\Theta}{xy}$ are consecutive letter pairs not contained in any factor of the factorization.
Also, $xy$ and $\Refl{\Theta}{xy}$ are not both in $\mathcal{R}$, a contradiction.

\textbf{Case 3: $\bm{|A| = 0, |C| > 0}$.}
Suppose, for the sake of contradiction and without loss of generality, that $B[1] = \Refl{\Theta}{C}[1]$.
Let $x = \Refl{\Theta}{B}[-1]$ and $y = B[1]$ and thus $B[-1] = \Refl{\Theta}{x}$ and $C[1] = \Refl{\Theta}{y}$.
So $xy$ and $\Refl{\Theta}{xy}$ are consecutive letter pairs not contained in any factor of the factorization and are not both in $\mathcal{R}$, a contradiction. 
\end{proof}

\begin{lemma}
\label{lem:T2R-reflection-big-gapped}
Let $W$ be a word.
There are $O(|W|)$ admissible factor pairs $X, \Refl{\Theta}{X}$ of $W$ such that $|X| \geq |W|/6$.
They can be enumerated in $O(|W|)$ time. 
\end{lemma} 

\begin{proof}
The idea is to search for admissible factor pairs according to the distance $d$ between $W[a] = X[1]$ and $W[a+d] = \Refl{\Theta}{X}[1]$, called the \emph{match distance}, for each $d$ with $|W|/6 \leq d \leq 5|W|/6$.
Because the factor pairs have $|X| \geq |W|/6$, any factor pair has $X$ containing $W[t]$ for some $t = i|W|/6$ with $i \in \{1, 2, \dots, 6\}$, called an \emph{outpost index}.
Fixing a match distance $d$ and outpost index $t$ uniquely determines an admissible factor pair defined by the largest interval $W[l..r]$ with $l \leq t \leq r$ such that $W[j] = \Refl{\Theta}{W[j+d]}$ for all $l \leq j \leq r$.
Thus there are $O(|W|)$ such factor pairs.

Enumerating pairs is done in $O(|W|)$ time using two steps.
First, preprocess $W$, $\Refl{\Theta}{W}$, $\Rev{W}$, and $\Refl{\Theta}{\Rev{W}}$ using Lemma~\ref{lem:HLF-longest-common-extension} to answer queries of the form:
\begin{itemize}
\item ``What is the longest common prefix of $W[t..|W|]$ and $\Refl{\Theta}{W}[t+d..|W|]$?''
\item ``What is the longest common suffix of $W[1..t]$ and $\Refl{\Theta}{W}[1..t+d]$?''
\end{itemize} 
Preprocessing takes $O(|W|)$ total time and enables $O(1)$-time queries. 

Next, iterate through all combinations of match distance $d$ and outpost index $t$, and for each compute $l, r$ using the data structure from the previous step.
Output the resulting admissible factor pair $X = W[l..r]$, $\Refl{\Theta}{X} = W[l+d..r+d]$ unless no interval exists, i.e., $W[t] \neq \Refl{\Theta}{W[t+d]}$.
The $O(|W|)$ combinations of $d$ and $t$ take $O(1)$ time each and $O(|W|)$ total time to compute.
\end{proof}

\subsection{Long admissible factors}

\begin{lemma}
\label{lem:T2R-presuf-period}
Let $W$ be a boundary word with a factor $X$.
Let $P,S \Mirror W$ with such that $P \Prefix X$, $S \Suffix X$, and $P \neq S$.
Then $X$ has a period of length $2|X|-(|P|+|S|)$.
\end{lemma}

\begin{proof}
Since $P$ and $S$ are mirror, there exists $X' \Factor W$ with $|X'| = |X|$, $\Back{P} \Prefix X'$, and $\Back{S} \Suffix X'$.
Observe that $X$ has a period of length $r \geq 1$ if and only if $X[i] = X[i+r]$ for all $1 \leq i \leq |X|-r$.
Let $1 \leq i \leq |P|+|S|-|X|$.
Then $1 \leq |P|+1-i \leq |X|$ and $1 \leq |P|+1+|\Back{S}|-|X'|-i \leq |\Back{S}|$.
So:
\begin{equation*}
\begin{split}
X[i] &= P[i] \\
&= \Comp{\Back{P}}[|P|+1-i] \\
&= \Comp{X'}[|P|+1-i] \\
&= \Comp{\Back{S}}[|P|+1+|\Back{S}|-|X'|-i] \\
&= \Comp{\Back{S}}[|\Back{S}|+1-(i+|X'|-|P|)] \\
&= S[i+|X'|-|P|] \\
&= X[i+|X'|-|P|+(|X|-|S|)] \\
&= X[i+2|X|-(|P|+|S|)]
\end{split}
\end{equation*}
Since $P \neq S$, $2|X|-(|P|+|S|) \geq 2|X|-(2|X|-1) = 1$.
So $X$ has a period of length $2|X|-(|P|+|S|)$.
\end{proof}

\begin{lemma}
\label{lem:T2R-long-middle-nonadmissible}
Let $W$ be a boundary word with $X \Factor W$.
Let $P, S \Mirror W$ such that $P \Prefix X$, $S \Suffix X$, and $P \neq S$.
Any factor $Y \Middle X$ with $|Y| > 2|X|-(|P|+|S|)$ is not an admissible factor of $W$.
\end{lemma}

\begin{proof}
By Lemma~\ref{lem:T2R-presuf-period}, $X$ has a period of length $r = 2|W|-(|P|+|S|)$.
Let $Y \Middle X$ and $|Y| > r$.

Let $X' \Factor W$ with $|X'| = |X|$ and the center of $X'$ exactly $|W|/2$ letters from the center of $X$.
Then $\Back{P} \Prefix X'$, $\Back{S} \Suffix X'$, and $\Back{Y} \Middle X'$.
Again by Lemma~\ref{lem:T2R-presuf-period}, $X'$ has a period of length $r$.

Let $U, V \Factor W$ such that $W = YU\Back{Y}V$.
Since $Y$ is a middle factor of $X$, $U[1..1]$ is in $X$.
Since $X$ has a period of length $r$ and $|Y| > r$, $U[1] = Y[|Y|+1-r] = \Comp{\Back{Y}[r]}$.
Since $\Back{Y}$ is a middle factor of $X'$ and $X'$ has a period of length $r$, $U[-1] = \Back{Y}[r]$.
So $U[1] = \Comp{U[-1]}$ and $Y$ is not admissible.
\end{proof}

The proof of the following result is nearly identical to that of Lemma~\ref{lem:QRT-long-palindrome-ends-constant}, due to similar prerequisite structural results of Lemma~\ref{lem:T2R-long-middle-nonadmissible} and~\ref{lem:QRT-middle-palins-nonadmissible}.

\begin{lemma}
\label{lem:T2R-constant-long-admissible-ends}
Let $W$ be a boundary word.
There exists a set $\mathscr{F}$ of $O(1)$ factors of $W$ such that every admissible $F \Mirror W$ with $|F| \geq |W|/6$ is an affix factor of an element of $\mathscr{F}$.
The set $\mathscr{F}$ can be enumerated in $O(|W|)$ time. 
\end{lemma}

\begin{proof}
\textbf{Three factors.}
Let $P_1, P_2, P_3 \Mirror W$ be admissible with $|P_1|, |P_2|, |P_3| \geq |W|/6$ and centers contained in a factor of $W$ with length at most $|W|/14$.

Let $X \Factor W$ be the shortest factor such that $P_1, P_2, P_3 \Factor X$.
For some $i, j \in \{1, 2, 3\}$, $P_i \Prefix X$ and $P_j \Suffix X$.
We prove that if $i \neq j$, then $P_1, P_2, P_3 \Affix X$.

Without loss of generality, suppose $i = 1$, $j = 2$ and so $P_3 \Middle X$.
By Lemma~\ref{lem:T2R-long-middle-nonadmissible}, since $P_3$ is admissible, $|P_3| \leq 2|X| -(|P_1|+|P_2|) \leq |P_1| + |W|/7 + |P_2| - (|P_1|+|P_2|) = |W|/7 < |W|/6$, a contradiction.
So $P_3 \Affix X$.

\textbf{More than three factors.}
Consider a set $\mathscr{I} = \{F_1, F_2, \dots, F_m\}$ of at least three admissible factors of $W$ of length at least $|W|/6$ such that the centers of the factors are contained in a common factor of $W$ of length $|W|/14$.
We will prove that every element of $\mathscr{I}$ is an affix factor of one of two factors of $W$.

Let $G \Factor W$ be the shortest factor such that $F_i \Factor G$ for every $F_i \in \mathscr{I}$.
It is either the case that there exist distinct $F_l, F_r \in \mathscr{I}$ with $F_l \Prefix G$, $F_r \Suffix G$, or that $G \in \mathscr{I}$ and every $F_i \in \mathscr{I}$ besides $G$ has $F_i \Middle G$.

In the first case, $F_i \Affix G$ for any $i \neq l, r$ by the previous claim regarding three factors.
Also $F_l, F_r \Affix G$.
So every factor in $\mathscr{I}$ is an affix factor of $G$.

In the second case, let $G' \Factor G$ be the shortest factor with the same center as $G$ such that every factor in $\mathscr{I}$ excluding $G$ is a factor of $G'$.
Clearly $G' \Mirror W$ and $G'$ is not admissible.
Without loss of generality, there exists $F_p \in \mathscr{I}$ such that $F_p \Prefix G'$.
Since $F_p$ is admissible and $G'$ is not, $F_p \neq G'$.

Applying Lemma~\ref{lem:T2R-long-middle-nonadmissible} with $X = G'$, $P = F_p$, $S = G'$, every middle factor of $G'$ in $\mathscr{I}$ has length at most $2|G'|-(|G'|+|F_p|) \leq |G'| - |F_p| \leq |W|/7 < |W|/6$.
So every factor of $G'$ in $\mathscr{I}$ is an affix factor of $G'$.
Thus every factor in $\mathscr{I}$ is either $G$ or an affix factor of $G'$.

\textbf{All factors.}
Partition $W$ into 15 factors $I_1, I_2, \dots, I_{15}$ each of length at most $|W|/14$.
Let $\mathscr{I}_i$ be the set of factors with centers containing letters in $I_i$.
Then by the previous claim regarding more than three factors, there exists a set $\mathscr{F}_i$ ($G$ and possibly $G'$) such that every element of $\mathscr{I}_i$ is an affix factor of an element of $\mathscr{F}_i$ and $|\mathscr{F}_i| \leq 2$.
So every admissible $F \Mirror W$ with $|F| \geq |W|/6$ is an affix factor of an element of $\mathscr{F} = \bigcup_{i=1}^{15}{\mathscr{F}_i}$ and $|\mathscr{F}| \leq 2\cdot15$.

\textbf{Efficient enumeration.}
Use Manacher's algorithm~\cite{Manacher-1975} to compute all admissible (also called \emph{maximal}) palindrome factors, eliminating those with length less than $|W|/6$.
Partition these factors into the~15 sets $I_1, \dots, I_{15}$.
For each set, sort the factors by both first and last letter, compute $G$ and (if defined) $G'$, and output them.
\end{proof}

\subsection{Algorithm}

\begin{lemma}
\label{lem:T2R-factors-admissible-2}
Let $P$ be a polyomino and $\Bou{P} = A B C \Back{A} \Refl{\Theta}{C} \Refl{\Theta}{B}$.
Then $A$, $\Back{A}$ are admissible.
\end{lemma}

\begin{proof}
Because of the factorization's structure, $A$ and $\Back{A}$ are mirror.
Suppose, for the sake of contradiction and without loss of generality, that $B[1] = \Comp{C[-1]}$.
Let $x = \Refl{\Theta}{B}[1]$.
So $\Refl{\Theta}{C}[-1] = \Refl{\Theta}{C[-1]} = \Refl{\Theta}{\Comp{B[1]}} = \Refl{\Theta}{\Comp{\Refl{\Theta}{x}}} = \Comp{x}$.
Thus $\Comp{x} x$ is a subword of $\Bou{P}$, a contradiction.
\end{proof}

\begin{theorem}
\label{thm:T2R-algorithm}
Let $P$ be a polyomino with $|\Bou{P}| = n$.
It can be decided in $O(n)$ time if $\Bou{P}$ has a type-2 reflection factorization.
\end{theorem}

\begin{proof}
Recall that a type-2 reflection factorization of $W$ has the form $W = A B C \Back{A} \Refl{\Theta}{C} \Refl{\Theta}{B}$.
Without loss of generality, either $|A| \geq |W|/6$ or $|B| \geq |W|/6$.

\textbf{Case 1: $\bm{|A| \geq |W|/6}$.}
The factors $A$ and $\Back{A}$ are admissible by Lemma~\ref{lem:T2R-factors-admissible-2}. 
Observe that each admissible factor has a distinct center $W[i..j]$ with $j \in \{i, i+1\}$.
Moreover, it is the longest factor $X = LR$ with this center such that the factor $\Back{X} = \Back{R} \Back{L}$ has center $W[i+|W|/2..j+|W|/2]$.
Use Lemma~\ref{lem:HLF-longest-common-extension} to preprocess $W$, $\Comp{W}$, $\Rev{W}$, and $\Back{W}$ in $O(|W|)$ time to then compute the admissible factor with each center in $O(|W|)$ total time.

Next, use Lemma~\ref{lem:T2R-constant-long-admissible-ends} to compute $\mathscr{F}$, a set of factors such that every factor $A$ with $|A| \geq |W|/6$ is an affix factor of element of $\mathscr{F}$.
This takes $O(|W|)$ time.
For each element $F \in \mathscr{F}$, the prefix admissible factors of $F$ are handled together in $O(|W|)$ time.
The suffix factors of each $F$ are handled similarly and symmetrically.
Since $|\mathscr{F}| = O(1)$, $O(|W|)$ total time is spent. 
The remainder of this case describes how the prefix admissible factors of a single factor $F$ are handled. 
The key is to enumerate and process each candidate factor pair $B, \Refl{\Theta}{B}$, rather than candidate factor pairs $A, \Back{A}$.

\begin{figure}[ht]
\centering
\includegraphics[scale=1.0]{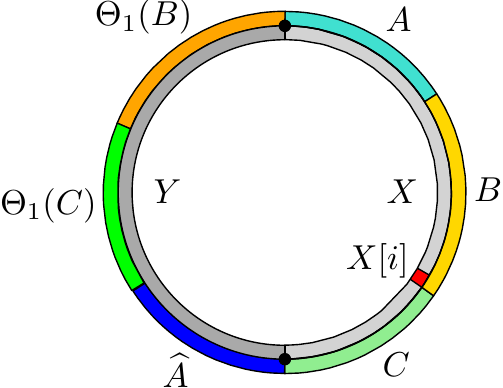}
\caption{The words used in the proof of Theorem~\ref{thm:T2R-algorithm}.}
\label{fig:T2R-algorithm}
\end{figure}

Let $W = XY$ such that $|Y| = \pm |X|$ and $F, X$ start at the same letter of $W$ (see Figure~\ref{fig:T2R-algorithm}).  
Reuse the data structure previously constructed by Lemma~\ref{lem:HLF-longest-common-extension} to support the following queries in $O(1)$ time: ``What is the length of the longest common suffix of $\Refl{\Theta}{Y}$ and $X[1..i]$?''.
For each $i$ with $1 \leq i \leq |X|$, denote the result of this query $s_i$.

Each pair of longest common suffixes $X[i+1-s_i..i], Y[|Y|+1-s_i..|Y|]$ is a candidate factor pair $B, \Refl{\Theta}{B}$.
Since $B, \Refl{\Theta}{B}$ are admissible by Lemma~\ref{lem:T2R-factors-admissible-1}, no non-longest suffix can be candidates. 
For each pair, verify that $A = X[1..i-s_i]$, the prefix of $X$ preceeding $B$, is an admissible prefix factor of $F$.
This takes $O(1)$ time after precomputing a look-up table of the admissible prefix factors of $F$.
Since $A$ is a mirror, the prefix of $Y$ with length $|A|$ is $\Back{A}$.
Finally, verify that $X[i+1..|X|] = \Refl{\Theta}{Y[|A|+1..|Y|-s_i]}$, i.e., that the remaining words before and after $\Back{A}$ are factors $C$ and $\Refl{\Theta}{C}$, respectively. 
This is done in $O(1)$ time by using the data structure constructed by Lemma~\ref{lem:HLF-longest-common-extension}.

If both verifications succeed, then the named factors form a type-2 reflection factorization of $W$. 
In total, $O(1)$ time is spent per choice of $i$, the index of $W$ where $B$ ends, and thus $O(|W|)$ time total.

\textbf{Case 2: $\bm{|B| \geq |W|/6}$.}
Compute a $O(|W|)$-sized superset of all factor pairs $B, \Refl{\Theta}{B}$ using Lemma~\ref{lem:T2R-reflection-big-gapped}.
For each pair, repeat the verifications done in case~1.
Using the same analysis as case~1, this takes $O(|W|)$ total time.
\end{proof}


\section{Type-1 Half-Turn-Reflection Factorizations}
\label{sec:T1H}

\begin{definition}
Let $W$ be a boundary word.
A \emph{type-1 half-turn-reflection factorization} of a boundary word $W$ has the form $W = A B C \Back{A} D \Refl{\Theta}{D}$ with $B$, $C$ palindromes.
\end{definition}

\begin{lemma}
\label{lem:T1H-factors-admissible}
Let $P$ be a polyomino and $\Bou{P} = A B C \Back{A} D \Refl{\Theta}{D}$ with $B$, $C$ palindromes.
Then reflect square $D \Refl{\Theta}{D}$, and palindromes $B$, $C$ are admissible.
\end{lemma}

\begin{proof}
\textbf{$\bm{D \Refl{\Theta}{D}}$ is admissible.}
Suppose, for the sake of contradiction and without loss of generality, that $\Refl{\Theta}{D}[1] = \Refl{\Theta}{A[1]}$.
Let $x = A[1]$.
Then $\Back{A}[-1] = \Comp{x}$ and $D[1] = \Refl{\Theta}{(\Refl{\Theta}{D}[1])} = \Refl{\Theta}{(\Refl{\Theta}{A[1]})} = A[1] = x$.
So $\Comp{x} x$ is a subword of $\Bou{P}$, a contradiction.

\textbf{$\bm{B}$, $\bm{C}$ are admissible.}
Consider the admissibility of $B$; the other case follows by symmetry.
Suppose, for the sake of contradiction, that $A[-1] = C[1]$.
Let $x = A[-1]$.
Then $\Back{A}[1] = \Comp{x}$ and $C[-1] = C[1] = x$.  
So $x \Comp{x}$ is a subword of $\Bou{P}$, a contradiction.
\end{proof}

\begin{theorem}
Let $P$ be a polyomino with $|\Bou{P}| = n$.
It can be decided in $O(n\log{n})$ time if $\Bou{P}$ has a type-1 half-turn-reflection factorization.
\end{theorem}

\begin{proof}
The algorithm has two phases.
Phase~1 is computing a $O(|W|\log{|W|})$-sized set of factors such that if any of these is can be factored into two admissible palindromes, then $W$ has a type-1 half-turn-reflection factorization.
Phase~2 is deciding if any of these factors has such a factorization.

\textbf{Phase 1.}
A type-1 half-turn-reflection factorization of $W$ has the form $W = A B C \Back{A} D \Refl{\Theta}{D}$ with $B$, $C$ palindromes.
The factor $D \Refl{\Theta}{D}$ is admissible by Lemma~\ref{lem:T1H-factors-admissible}.
Use Lemma~\ref{lem:T1R-small-fast-superset} to enumerate a $O(|W|\log{|W|})$-sized superset of all admissible factors $D \Refl{\Theta}{D}$.

Preprocess $W$ and $\Back{W}$ using Lemma~\ref{lem:HLF-longest-common-extension} to answer the following queries in $O(1)$ time: ``What is the longest admissible pair $A, \Back{A}$ with $A$ ending at $W[i]$ and $\Back{A}$ starting at $W[j]$?''.
For each candidate $D \Refl{\Theta}{D} = W[l..r]$, compute the longest admissible pair $A, \Back{A}$ with $A$ ending at $W[l-1]$ and $\Back{A}$ starting at $W[r+1]$.
This partial factorization $A D \Refl{\Theta}{D} \Back{A}$ can be completed into a type-1 half-turn-reflection factorization if and only if $W[r+1+|A|..l-1-|A|]$ has a factorization into two palindromes.
Call the set of such factors computed for all candidate $D \Refl{\Theta}{D}$ \emph{completion factors}.
These factors take $O(|W|\log{|W|})$ total time to compute.

\textbf{Phase 2.}
The factors $B$, $C$ of a type-1 half-turn-reflection factorization are admissible by Lemma~\ref{lem:T1H-factors-admissible}.
Use Manacher's algorithm~\cite{Manacher-1975} to compute the $O(|W|)$-sized set of all admissible factors.
For each letter in $W$, create length-sorted lists of the admissible palindrome factors that start and end at this letter.
Constructing these lists takes $O(|W|)$ time using counting sort, since there are $O(|W|)$ factors of integer lengths between~1 and~$|W|$.
Construct similar lists for the completion factors from phase~1 in $O(|W|\log{|W|})$ time.

By Lemma~\ref{lem:HLF-extremal-double-palindromes}, a completion factor has a factorization into two admissible palindrome factors if and only if such a factorization exists utilizing the longest admissible prefix or suffix palindrome of the completion factor.
For each letter in $W$, use a two-finger scan to compute the longest admissible prefix palindrome for each completion factor starting at this letter.
For each pair, verify that the remainder of the factor is also a admissible palindrome factor using a precomputed look-up table.
If so, then a type-1 half-turn-reflection factorization has been found.
Do a symmetric scan through the completion factors and admissible palindrome factors ending at the letter. 

Scans take time linear in the number of admissible palindromes and completion factors involved.
Since each factor appears in two scans (one each for the letter the factor starts and ends at) the total time taken is $O(|W|\log{|W|})$. 
\end{proof}


\section{Type-2 Half-Turn-Reflection Factorizations}
\label{sec:T2H}

\begin{definition}
A \emph{type-2 half-turn-reflection factorization} of a boundary word $W$ has the form $W = A B C D \Refl{\Theta}{B} \Refl{\Phi}{D}$ with $A$, $C$ palindromes and $\D{\Theta}-\D{\Phi} = \pm \D{90}$.
\end{definition}

\begin{lemma}
\label{lem:T2H-subroutine-1}
Let $U$, $V$ be words.
Let the prefix and suffix palindrome factorizations of $U$ and $V$ be given, along with a data structure supporting $O(1)$-time longest common prefix queries for subwords of $V$ and $\Refl{\Theta}{U}$.
It can be decided in $O(\log{|U|} + \log{|V|})$ time if there exist factorizations $U = \Refl{\Theta}{D} A$, $V = C D$ with $A$, $C$ palindromes. 
\end{lemma}

\begin{proof}
Solve by considering solutions such that $C$ ends with each of the $O(\log{|V|})$ repeated factors $X_i^{r_i}$ in the prefix palindrome factorization of $V$.
If $r_i \leq 2$, then check if either choice of $C$ ending at one of the factors leaves a factor $D \Suffix V$ such that $\ORefl{\Phi}(D)$ is a prefix of $U$ and the suffix of $U$ with length $|U|-|D|$ is a palindrome.
Both checks together take $O(\log{|U|})$ time by using a single longest common prefix query and a scan of the suffix palindrome factorization of $U$.

Otherwise $r_i \geq 3$.
Compute the largest $k$ such that the longest common prefix of $X_i^{r_i}$ and $\Refl{\Phi}{U}$ is at least $k|X_i|$ using a longest common prefix query. 
If $k = 0$, then no factorization $V = C D$ exists using a factor $D$ starting at a repetition of $X_i$; so assume that $k \geq 1$.
Let $Q$ be the suffix of $V$ following $X_i^{r_i}$.
By the maximality of $r_i$ that follows from the definition of a prefix palindrome factorization, either: (1)~$|Q| \geq |X_i|$ and $X_i \not \Prefix Q$ or (2)~$|Q| < |X_i|$.
Also, in both cases, no factorization $V = C D$ exists such that $D$ has a prefix $X_i^{k+1}$ or, equivalently, such that $D$ starts at the $j$th repetition of $X_i$, for any $j \leq r_i - k$. 

\textbf{Case 1: $\bm{|Q| \geq |X_i|}$ and $\bm{X_i \not \Prefix Q}$.} 
In this case, $D$ also cannot start at the $j$th repetition of $X_i$ for any $j > r_i - k + 1$ and so $D$ must start at the repetition $r_i - k + 1$ of $X_i^{r_i}$.
Thus the factors $C$, $D$, and $\Refl{\Phi}{D}$ are fixed.
Search for a suffix factor of $U$ with length $|U|-|D|$ by scanning the suffix palindrome factorization of $U$.

\textbf{Case 2: $\bm{|Q| < |X_i|}$.}
Check if the longest common prefix of $Q$ and $\Refl{\Phi}{U}$ is $Q$, i.e., if $\Refl{\Phi}{Q}$ is a prefix of $U$. 
If not, then $D$ cannot start the $j$th repetition for any $j > r_i-k+1$ since $\Refl{\Phi}{X_i^k}$ is a prefix of $U$.
Thus $D$ must start at repetition $r_i-k+1$; search for a factor $A$ to complete the factorizations in $O(\log{|U|})$ time as before. 

Otherwise $Q$ is a prefix of $X_i$ and thus $D$ can start at the $j$th repetition of $X_i$ for all $j \geq r_i-k+1$.
Check for the existence of a choice of $C$, $D$ such that $U = \Refl{\Phi}{D} A$ and $A$ a palindrome by scanning the suffix palindrome factorization of $U$ for $A$ with $|A| \in \{ |U| - ((r_i-j+1)|X_i| + |Q|) : j \geq r_i-k+1 \}$.
Since the lengths of the palindromes for each repeated factor also form a linear set, this takes $O(1)$ per repeated factor and $O(\log{|U|})$ total time.

\textbf{Running time.}
Each repeated factor except the last has either $r_i \leq 2$ or case~1 applies.
Moreover, each can be handled in $O(1)$ time plus a search in the suffix palindrome factorization of $U$ for a palindrome of length $|U|-|D|$.
As $i$ increases, the length of the suffix palindrome searched for decreases, so these can be handled together in $O(\log{|U|})$ total time.
Thus all repeated factors except (possibly) the last can be handled in time $O(\log{|V|} + \log{|U|})$, while the optional handling of case~2 takes $O(\log{|U|})$ time. 

\end{proof}

\begin{lemma}
\label{lem:T2H-subroutine-2}
Let $B$, $U$ be words.
Let the prefix and suffix palindrome factorizations of $U$ be given, along with a data structure supporting $O(1)$-time longest common prefix queries for $B$ and subwords of $U$.
It can be decided in $O(\log{|U|})$ time if there exists a factorization $U = A B C$ with $A$, $C$ palindromes.
\end{lemma}

\begin{proof}
Solve by considering solutions such that~$B$ starts with each of the $O(\log{|U|})$ repeated factors $X_i^{r_i}$ in the prefix palindrome factorization of $V$.
If $r_i \leq 2$, then check if either choice of $B$ starting at one of the repetitions of $X_i$ leaves a remainder of $U$, called $C$, that is a palindrome.
Both checks together take $O(\log{|U|})$ time by using a single longest common prefix query and a scan of the suffix palindrome factorization of $U$.

Otherwise $r_i \geq 3$.
Compute ${\rm lcp}(X_i^{r_i}, B)$, the longest common prefix of $X_i^{r_i}$ and $B$.
Let $k$ be the largest integer such that ${\rm lcp}(X_i^{r_i}, B) \geq k|X_i|$.
Either ${\rm lcp}(X_i^{r_i}, B) < |B|$ is less than or equal to $|B|$.

\textbf{Case 1: lcp$\bm{(X_i^{r_i}, B) < |B|}$.}
No solution $B$ can start before repetition $r_i-k+1$, since $B$ is not a prefix of the remaining word.
Also, no solution $B$ can start after repetition $r_i-k+1$ by the maximality of $r_i$.
So check for a solution with $B$ starting at repetition $r_i-k+1$ in $O(\log{|U|})$ time by searching in $O(\log{|U|})$ time for a suffix palindrome $C$ of length $|U|-|AB|$ in the suffix palindrome factorization of $U$. 

\textbf{Case 2: lcp$\bm{(X_i^{r_i}, B) = |B|}$.}
Observe $B$ is a prefix of the suffixes of $U$ starting at repetitions $1, 2, \dots, r_i-k+1$.
Moreover, $B$ is not a prefix of any suffix of $U$ starting at a repetition after $r_i-k+1$ by the maximality of $r_i$.
The lengths of the suffixes following those choices of $B$ starting at repetitions $1, 2, \dots r_i-k+1$ form a linear set.
Check for a solution $C$ with length in this linear set.
This check is possible in $O(\log{|U|})$ time, since each repeated factor in the suffix palindrome factorization can be checked in $O(1)$ time by solving a system of two linear equations.
 
\textbf{Amortized $O(1)$-time search for $\bm{C}$.}
In both of the above cases, the length of $C$ completing a solution factorization is equal to $|U|-|AB|$, where $|B|$ is fixed and $|A|$ is a prefix ending at a repetition of $X_i$.
So the lengths of $C$ strictly decrease as $i$ increases. 
Maintain a pointer into the suffix palindrome factorization of $U$ for searching for $C$ across repeated factors of the prefix palindrome factorization of $U$.
This pointer visits each factor once, and thus the searches take $O(\log{|U|})$ total time and $O(1)$ amortized time each.
So in total $O(1)$ amortized time is spent per $i$, and thus $O(\log{|U|})$ time total.
 

\end{proof}

\begin{lemma}
\label{lem:T2H-factors-admissible}
Let $P$ be a polyomino and $\Bou{P} = A B C D \Refl{\Theta}{B} \Refl{\Phi}{D}$ with $A$, $C$ palindromes.
Then gapped reflect squares $B, \Refl{\Theta}{B}$ and $D, \Refl{\Phi}{D}$, and palindromes $A$, $C$ are admissible.
\end{lemma}

\begin{proof}
\textbf{$\bm{A}$, $\bm{C}$ are admissible.}
Since $A$ and $C$ are symmetric, proving the claim for $A$ suffices.
Suppose, for the sake of contradiction and without loss of generality, that $\Refl{\Phi}{D}[-1] = B[1]$.
Let $x = \Refl{\Theta}{B}[1]$.
Then $D[-1] = \Refl{\Phi}{\Refl{\Theta}{x}} = \Comp{x}$.
So $\Comp{x} x$ is a subword of $W$, a contradiction.

\textbf{$\bm{B, \Refl{\Theta}{B}}$ are admissible.}
Suppose, for the sake of contradiction and without loss of generality, that $\Refl{\Theta}{C[1]} = \Refl{\Phi}{D}[1]$.
Let $y = C[1]$. 
Then $\Refl{\Phi}{D}[1] = \Refl{\Theta}{y}$ and thus $D[1] = \Refl{\Phi}{\Refl{\Theta}{y}} = \Comp{y}$.
Also, $y = C[1] = C[-1]$.
So $y \Comp{y}$ is a subword of $W$, a contradiction. 

\textbf{$\bm{D, \Refl{\Phi}{D}}$ are admissible.}
Suppose, for the sake of contradiction and without loss of generality, that $A[1] = \Refl{\Phi}{\Refl{\Theta}{B}[1]}$.
So $A[1] = \Comp{B[1]}$.
Then if $z = A[1] = A[-1]$, $z \Comp{z}$ is a subword of $W$, a contradiction.
\end{proof}

\begin{theorem}
Let $P$ be a polyomino with $|\Bou{P}| = n$.
It can be decided in $O(n\log{n})$ time if $\Bou{P}$ has a type-2 half-turn-reflection factorization.
\end{theorem}

\begin{proof}
Recall that a type-2 half-turn-reflection factorization of $W$ has the form $W = A B C D \Refl{\Theta}{B} \Refl{\Phi}{D}$ with $A$, $C$ palindromes and $\D{\Theta} - \D{\Phi} = \pm \D{90}$.
Without loss of generality, an element from $\{A, B, C, D\}$ has length at least $|W|/6$.
The case of $|C| \geq |W|/6$ is symmetric to that of $|A| \geq |W|/6$; the remaining three cases are each handled in $O(|W|\log{|W|})$ time.

All three cases use $O(1)$-time queries of the following form: ``For $Y, Z \in W, \Refl{\Theta}{W}, \Refl{\Phi}{W}, \Rev{W}$, what is the longest common prefix starting at $Y[i]$ and $Z[j]$?''.
Such queries are possible after $O(|W|)$ preprocessing by Lemma~\ref{lem:HLF-longest-common-extension}.
The three cases also use prefix and suffix palindrome factorizations of subwords of $W$.
These can be precomputed in $O(|W|\log{|W|})$ time by Lemma~\ref{lem:HLF-all-prefix-facts-fast}. 
The case of $|A| \geq |W|/6$ is reducible to the subroutines used to solve the other two cases.
For this reason, the other cases are considered first.

\textbf{$\bm{|B| \geq |W|/6}$.}
The pair of factors $B, \Refl{\Theta}{B}$ are admissible by Lemma~\ref{lem:T2H-factors-admissible}.
Enumerate a $O(|W|)$-sized superset of all pairs $B, \Refl{\Theta}{B}$ in $O(|W|)$ time using Lemma~\ref{lem:T2R-reflection-big-gapped}.
Determining if a given pair $B, \Refl{\Theta}{B}$ can be completed into a type-2 half-turn-reflection factorization is equivalent to solving an instance of the following problem: ``Given words $U, V$ and their prefix and suffix palindrome factorizations, do there exist factorizations $U = \Refl{\Phi}{D} A$, $V = C D$ with $A$, $C$ palindromes?''.
Each such instance is solved in $O(\log{|U|} + \log{|V|})$ time by Lemma~\ref{lem:T2H-subroutine-1}.
Thus the entire case is handled in $O(|W|(\log{|U|} + \log{|V|})) = O(|W|\log{|W|})$ time.

\textbf{$\bm{|D| \geq |W|/6}$.}
As in the case of $|B| \geq |W|/6$, enumerate a $O(|W|)$-sized superset of all $D, \Refl{\Phi}{D}$ in $O(|W|)$ time.
Handle each pair $D, \Refl{\Phi}{D}$ individually as an instance of the following problem: ``Given words $U$, $B$ and the prefix and suffix palindrome factorizations of $U$, does there exist a factorization $U = A B C$ with $A$, $C$ palindromes?''.
Each such instance is solved in $O(\log{|U|})$ time by Lemma~\ref{lem:T2H-subroutine-2}.
Thus the entire case is handled in $O(|W|\log{|U|}) = O(|W|\log{|W|})$ time.

\textbf{$\bm{|A| \geq |W|/6}$.}
First, Use Lemma~\ref{lem:QRT-long-palindrome-ends-constant} to compute $\mathscr{F}$, a $O(1)$-sized set of factors of $W$ such that every factor $A$ is an affix factor of element of $\mathscr{F}$.

Next, process the prefix palindromes of each $F \in \mathscr{F}$.
The first letter $W[f]$ of each $F$ defines a last letter $W[f-1]$ for $\Refl{\Phi}{D}$.
Combine this last letter of $\Refl{\Phi}{D}$ with the $O(|W|)$ choices for the last letter of $D$ to obtain $O(|W|)$ pairs $D, \Refl{\Phi}{D}$.
Compute each pair $D, \Refl{\Phi}{D}$ in $O(1)$ time using the previously computed suffix tree data structure.
Now handle each pair in $O(\log{|W|})$ time exactly as done for the case of $|D| \geq |W|/6$.
The $O(|W|)$ pairs are handled in $O(|W|\log{|W|})$ time total.

Finally, process the suffix palindromes of each $F \in \mathscr{F}$.
Each $F[-1] = W[f]$ defines the first letter $W[f+1]$ of $B$.
Combine this first letter of $B$ with the $O(|W|)$ choices for the first letter of $\Refl{\Theta}{B}$ to obtain $O(|W|)$ pairs $B, \Refl{\Theta}{B}$.
Compute each pair $B, \Refl{\Theta}{B}$ in $O(1)$ time and handle them in $O(\log{|W|})$ total time as done in the case of $|B| \geq |W|/6$.
\end{proof}

\section{Conclusion}

This work demonstrates that not just polynomial, but quasilinear-time algorithms exist for deciding tiling properties of a polyomino.
It remains to be seen if a linear-time algorithm exists, or whether a super-linear lower bound for one of the factorization forms exists.
The slowest algorithm is for half-turn factorizations, so it seems natural to attack this special case first. 

\begin{openproblem}
Can it be decided in $o(n\log^2{n})$-time if a polyomino $P$ with $|\Bou{P}| = n$ has a half-turn factorization?
\end{openproblem}

\begin{openproblem}
Can it be decided in $O(n)$-time if a polyomino $P$ with $|\Bou{P}| = n$ has an isohedral tiling of the plane? 
\end{openproblem}

For monohedral tilings containing only translations of the prototile, a polyomino has such a tiling only if it has one that is also isohedral~\cite{Beauquier-1991,Wijshoff-1984}.
Does this remain true for tilings using other sets of transformations of the prototile?
Modifying the anisohedral tile of Heesch~\cite{Heesch-1935} (see~\cite{Grunbaum-1980}) proves that the answer is ``no'' for tilings with reflected tiles, while an example of Rhoads~\cite{Rhoads-2005} proves that the answer is ``no'' for tilings with $\D{90}$ rotations of tiles.
This leaves one possibility open:

\begin{openproblem}
Does there exist a polyomino $P$ that has a tiling containing only translations and $\D{180}$ rotations of $P$ and every such tiling is anisohedral? 
\end{openproblem} 

As mentioned in Section~\ref{sec:overview}, there are isohedral tiling types (characterized by boundary factorizations) that cannot be realized by polyominoes due to angle restrictions.
Moreover, the boundary factorization forms here also apply to general polygons, under appropriate definitions of ``boundary word''.
Extending the algorithms presented here to polygons, along with developing algorithms for the remaining boundary factorizations is a natural goal.
However, significant challenge remains in efficiently converting a polygon's boundary into a word that can be treated with the approach used here. 

\begin{openproblem}
Can it be decided in $O(n\log^2{n})$ time if a polygon with $n$ vertices has an isohedral tiling of the plane? 
\end{openproblem}

Observe that pairs of tiles in a tiling that can be mapped to each other via a symmetry of the tiling induces a partition of the tiles.
Define a tiling to be $k$-isohedral if the partition has $k$ parts, e.g., an isohedral tiling is 1-isohedral.
Thus $k$-isohedral tilings are a natural generalization of isohedral tilings that allow increasing complexity; specifically, they cannot be characterized by a single boundary factorization.
A natural generalization of the problem considered here is as follows:



\begin{openproblem}
Can it be decided efficiently if a polyomino has a $k$-isohedral tiling?
\end{openproblem}

An approach described by Joseph Myers~\cite{Myers-2015} achieves a running time of approximately $n^{O(k^2)}$, though a precise analysis of the running time has not been performed.
A fixed-parameter tractable algorithm also may be possible.
On the other hand, a proof of \ccNP-hardness is unlikely, since it implies, for each $c \in \mathbb{N}$, the existence of prototiles whose only tilings are $k$-isohedral for $k \geq c$.
Such tiles are only known to exist for $c \leq 10$~\cite{Myers-2015}.

\section*{Acknowledgements}

The authors wish to thank anonymous reviewers for comments that improved the correctness of the paper.

\bibliographystyle{plain}
\bibliography{isohedral}

\begin{thebibliography}{10}

\bibitem{Alon-1997}
N.~Alon, R.~Yuster, and U.~Zwick.
\newblock Finding and counting given length cycles.
\newblock {\em Algorithmica}, 17(3):209--223, 1997.

\bibitem{Apostolico-1995}
A.~Apostolico, D.~Breslauer, and Z.~Galil.
\newblock Parallel detection of all palindromes in a string.
\newblock {\em Theoretical Computer Science}, 141:163--173, 1995.

\bibitem{Beauquier-1991}
D.~Beauquier and M.~Nivat.
\newblock On translating one polyomino to tile the plane.
\newblock {\em Discrete \& Computational Geometry}, 6:575--592, 1991.

\bibitem{Brlek-2009b}
S.~Brlek, M.~Koskas, and X.~Proven\c{c}al.
\newblock A linear time and space algorithm for detecting path intersection.
\newblock In {\em DGCI 2009}, volume 5810 of {\em LNCS}, pages 397--408.
  Springer Berlin Heidelberg, 2009.

\bibitem{Brlek-2009a}
S.~Brlek, J.-M. X.~Proven\c{c}al, and F\'{e}dou.
\newblock On the tiling by translation problem.
\newblock {\em Discrete Applied Mathematics}, 157:464--475, 2009.

\bibitem{Crochemore-1994}
M.~Crochemore and W.~Rytter.
\newblock {\em Text Algorithms}.
\newblock Oxford University Press, 1994.

\bibitem{Daurat-2005}
A.~Daurat and M.~Nivat.
\newblock Salient and reentrant points of discrete sets.
\newblock {\em Discrete Applied Mathematics}, 151:106--121, 2005.

\bibitem{Fine-1965}
N.~J. Fine and H.~S. Wilf.
\newblock Uniqueness theorems for periodic functions.
\newblock {\em Proceedings of the American Mathematical Society}, 16:109--114,
  1965.

\bibitem{Galil-1978}
Z.~Galil and J.~Seiferas.
\newblock A linear-time on-line recognition algorithm for ``{P}alstar''.
\newblock {\em Journal of the ACM}, 25(1):102--111, 1978.

\bibitem{Gambini-2007}
L.~Gambini and L.~Vuillon.
\newblock An algorithm for deciding if a polyomino tiles the plane by
  translations.
\newblock {\em RAIRO - Theoretical Informatics and Applications},
  41(2):147--155, 2007.

\bibitem{Gardner-1975}
M.~Gardner.
\newblock More about tiling the plane: the possibilities of polyominoes,
  polyiamonds, and polyhexes.
\newblock {\em Scientific American}, pages 112--115, August 1975.

\bibitem{GoodmanStrauss-2000}
C.~Goodman-Strauss.
\newblock Open questions in tiling.
\newblock Online, published 2000.
\newblock http://comp.uark.edu/~strauss/papers/survey.pdf.

\bibitem{GoodmanStrauss-2010}
C.~Goodman-Strauss.
\newblock Can't decide? undecide!
\newblock {\em Notices of the American Mathematical Society}, 57(3):343--356,
  2010.

\bibitem{Grunbaum-1977}
B.~Gr\"{u}nbaum and G.~C. Shephard.
\newblock The eighty-one types of isohedral tilings in the plane.
\newblock {\em Mathematical Proceedings of the Cambridge Philosophical
  Society}, 82(2):177--196, 1977.

\bibitem{Grunbaum-1978}
B.~Gr\"{u}nbaum and G.~C. Shephard.
\newblock Isohedral tilings of the plane by polygons.
\newblock {\em Commentarii Mathematici Helvetici}, 53(1):542--571, 1978.

\bibitem{Grunbaum-1980}
B.~Gr\"{u}nbaum and G.~C. Shephard.
\newblock Tilings with congruent tiles.
\newblock {\em Bulletin of the American Mathematical Society}, 3:951--973,
  1980.

\bibitem{Gusfield-1997}
D.~Gusfield.
\newblock {\em Algorithms on Strings, Trees, and Sequences: Computer Science
  and Computational Biology}.
\newblock Cambridge University Press, 1997.

\bibitem{Heesch-1935}
H.~Heesch.
\newblock Aufbau der ebene aus kongruenten bereichen.
\newblock {\em Nachrichten von der Gesellschaft der Wissenschaften zu
  G\"{o}ttingen}, 1:115--117, 1935.

\bibitem{Heesch-1963}
H.~Heesch and O.~Kienzle.
\newblock {\em Fl\"{a}chenschluss: System der Formen l\"{u}ckenlos
  aneinanderschliessender Flachteile}.
\newblock Springer, 1963.

\bibitem{Hilbert-1902}
D.~Hilbert.
\newblock Mathematical problems.
\newblock {\em Bulletin of the American Mathematical Society}, 8(10):437--479,
  1902.

\bibitem{I-2014}
T.~I, S.~Sugimoto, S.~Inenaga, H.~Bannai, and M.~Takeda.
\newblock Computing palindromic factorizations and palindromic covers on-line.
\newblock In A.~S. Kulikov, S.~O. Kuznetsov, and P.~Pevzner, editors, {\em CPM
  2014}, volume 8486 of {\em LNCS}, pages 150--161. Springer, Switzerland,
  2014.

\bibitem{Keating-1999}
K.~Keating and A.~Vince.
\newblock Isohedral polyomino tiling of the plane.
\newblock {\em Discrete and Computational Geometry}, 21(4):615--630, 1999.

\bibitem{Kershner-1968}
R.~B. Kershner.
\newblock On paving the plane.
\newblock {\em The American Mathematical Monthly}, 75(8):839--844, 1968.

\bibitem{Knuth-1977}
D.~E. Knuth, J.~H. Morris, and V.~R. Pratt.
\newblock Fast pattern matching in strings.
\newblock {\em SIAM Journal on Computing}, 6(2):323--350, 1977.

\bibitem{Main-1984}
M.~G. Main and R.~J. Lorentz.
\newblock An ${O}(n\log{n})$ algorithm for finding all repetitions in a string.
\newblock {\em Journal of Algorithms}, 5:422--432, 1984.

\bibitem{Manacher-1975}
G.~K. Manacher.
\newblock A new linear-time ``on-line'' algorithms for finding the smallest
  initial palindrome of a string.
\newblock {\em Journal of the ACM}, 22(3):346--351, 1975.

\bibitem{Mann-2015}
C.~Mann, J.~McLoud-Mann, and D.~Von Derau.
\newblock Convex pentagons that admit $i$-block transitive tilings.
\newblock Technical report, arXiv, 2015.
\newblock http://arxiv.org/abs/1510.01186.

\bibitem{Matsubara-2009}
W.~Matsubara, S.~Inenaga, A.~Ishino, A.~Shinohara, T.~Nakamura, and
  K.~Hashimoto.
\newblock Efficient algorithms to compute compressed longest common substrings
  and compressed palindromes.
\newblock {\em Theoretical Computer Science}, 410:900--913, 2009.

\bibitem{Myers-2015}
J.~Myers.
\newblock Polyomino, polyhex, and polyiamond tiling.
\newblock Online, updated February 2012.
\newblock http://www.polyomino.org.uk/mathematics/polyform-tiling/.

\bibitem{Provencal-2008}
X.~Proven\c{c}al.
\newblock {\em Combinatoire des mots, g\'{e}om\'{e}trie discr\`{e}te et
  pavages}.
\newblock PhD thesis, Universit\'{e} du Qu\'{e}bec \`{a} Montr\'{e}al, 2008.

\bibitem{Reinhardt-1928}
K.~Reinhardt.
\newblock Zur zerlegung der euklidischen räume in kongruente polytope.
\newblock {\em Sitzungsberichte der Preussischen Akademie der Wissenschaften},
  pages 150--155, 1928.

\bibitem{Rhoads-2005}
G.~C. Rhoads.
\newblock Planar tilings polyominoes, polyhexes, and polyiamonds.
\newblock {\em Journal of Copmutational and Applied Mathematics}, 174:329--353,
  2005.

\bibitem{Schattschneider-1980}
D.~Schattschneider.
\newblock Will it tile? try the {C}onway criterion!
\newblock {\em Mathematics Monthly}, 53(4):224--233, 1980.

\bibitem{Schattschneider-1990}
D.~Schattschneider.
\newblock {\em Visions of Symmetry: Notebooks, Periodic Drawings, and Related
  Work of M. C. Escher}.
\newblock W. H. Freeman and Company, 1990.

\bibitem{Socolar-2011}
J.~E.~S. Socolar and J.~M. Taylor.
\newblock An aperiodic hexagonal tile.
\newblock {\em Journal of Combinatorial Theory, Series A}, 118(8):2207--2231.

\bibitem{Wijshoff-1984}
H.~A.~G. Wijshoff and J.~van Leeuwen.
\newblock Arbitrary versus periodic storage schemes and tessellations of the
  plane using one type of polyomino.
\newblock {\em Information and Control}, 62:1--25, 1984.

\bibitem{Winslow-2015}
A.~Winslow.
\newblock An optimal algorithm for tiling the plane with a translated
  polyomino.
\newblock In {\em 26th International Symposium on Algorithms and Computation
  (ISAAC)}, 2015.

\end{thebibliography}

\end{document}